\documentclass[conference,10pt,final]{IEEEtran}

\usepackage{amsmath, amssymb}
\usepackage{hyperref}

\usepackage{xcolor}
\usepackage[linesnumbered,algoruled,vlined]{algorithm2e}
\usepackage{xparse} 
\usepackage{array,multirow, tabularx}
\usepackage{graphicx}
\usepackage{enumitem}
\usepackage{comment}
\usepackage{stmaryrd}
\usepackage{url}
\usepackage{fixltx2e}
\usepackage{enumitem} 
\setlist[description]{%
  topsep=0pt,               
  itemsep=0pt,               
  labelwidth=0.0cm,
  labelindent=0pt,
  leftmargin=12pt
}
\setlist[enumerate]{%
  topsep=0pt,               
  itemsep=0pt,               
  labelwidth=0.0cm,
  labelindent=4pt,
  leftmargin=*
}
\setlist[itemize]{%
  topsep=0pt,               
  itemsep=0pt,               
  labelwidth=0.0cm,
  labelindent=4pt,
  leftmargin=*
}

\usepackage{amsthm}
\usepackage{thmtools,thm-restate}

 \theoremstyle{plain}
 \newtheorem{theorem}{Theorem}
 \newtheorem{lemma}[theorem]{Lemma}

 \theoremstyle{definition}
 \newtheorem{definition/}[theorem]{Definition}
 
 \newenvironment{definition}
   {
   
    \pushQED{\qed}\begin{definition/}
   }{
   \popQED\end{definition/}
   
   }

 \theoremstyle{remark}

\newtheorem{observation}[theorem]{Observation}

\newcommand{\req}{r}

\newcommand{\type}{\ensuremath{\tau}}

\newcommand{\VG}[1][\req]{\ensuremath{G_{#1}}}
\newcommand{\VV}[1][\req]{\ensuremath{V_{#1}}}
\newcommand{\VE}[1][\req]{\ensuremath{E_{#1}}}
\newcommand{\Vpaths}[1][\req]{\ensuremath{\mathcal{P}_r}}

\NewDocumentCommand{\VGP}{O{\req} O{i} O{j}}{\ensuremath{G^{#2,#3}_{#1}}}

\newcommand{\Vsource}[1][\req]{\ensuremath{o^+_{\req}}}
\newcommand{\Vsink}[1][\req]{\ensuremath{o^-_{\req}}}

\newcommand{\Vcap}[1][\req]{\ensuremath{c_{#1}}}
\newcommand{\Vlat}[1][\req]{\ensuremath{l_{#1}}}
\newcommand{\Slat}[1][S]{\ensuremath{l_{#1}}}
\newcommand{\VVloc}[1][i]{\ensuremath{V^{{#1}}_{S}}}
\newcommand{\VVlocForbidden}[1][i]{\ensuremath{\overline{V}^{{#1}}_{S}}}
\newcommand{\VEloc}[1][i,j]{\ensuremath{E^{{#1}}_{S}}}
\newcommand{\VElocForbidden}[1][i,j]{\ensuremath{\overline{E}^{{#1}}_{S}}}

\newcommand{\SG}{\ensuremath{G_S}}

\newcommand{\SV}{\ensuremath{V_S}}
\newcommand{\SE}{\ensuremath{E_S}}

\newcommand{\Scap}{\ensuremath{c_{S}}}

\newcommand{\map}[1][\req]{\ensuremath{m}}
\newcommand{\mapV}[1][\req]{\ensuremath{m_V}}
\newcommand{\mapE}[1][\req]{\ensuremath{m_E}}

\makeatletter
\newcommand{\nosemic}{\renewcommand{\@endalgocfline}{\relax}}
\newcommand{\dosemic}{\renewcommand{\@endalgocfline}{\algocf@endline}}
\newcommand{\popline}{\Indm\dosemic}
\let\oldnl\nl
\newcommand{\nonl}{\renewcommand{\nl}{\let\nl\oldnl}}
\makeatother

\makeatletter
\newcommand{\removelatexerror}{\let\@latex@error\@gobble}
\makeatother

\newcounter{ipCounter}
\NewDocumentEnvironment{IPFormulation}{m}{%
\refstepcounter{ipCounter}
\begin{algorithm}[#1]%

}{%
\end{algorithm}
\addtocounter{algocf}{-1}
}

\NewDocumentEnvironment{IPFormulationStar}{m}{%
\refstepcounter{ipCounter}
\begin{algorithm*}[#1]%

}{%
\end{algorithm*}
\addtocounter{algocf}{-1}
}

\DeclareDocumentCommand{\NodeG}{O{\req} O{\pi}}{\ensuremath{G^N_{#1,#2}}}
\DeclareDocumentCommand{\NodeV}{O{\req} O{\pi}}{\ensuremath{V^N_{#1,#2}}}
\DeclareDocumentCommand{\NodeE}{O{\req} O{\pi}}{\ensuremath{E^N_{#1,#2}}}

\DeclareDocumentCommand{\EdgeG}{O{\req} O{i} O{j} O{u}}{\ensuremath{G^E_{#1,#2,#3,#4}}}
\DeclareDocumentCommand{\EdgeV}{O{\req} O{i} O{j} O{u}}{\ensuremath{V^E_{#1,#2,#3,#4}}}
\DeclareDocumentCommand{\EdgeE}{O{\req} O{i} O{j} O{u}}{\ensuremath{E^E_{#1,#2,#3,#4}}}

\DeclareDocumentCommand{\VESD}{O{\req}}{\ensuremath{\overrightarrow{E}_{#1}}}
\DeclareDocumentCommand{\VEOD}{O{\req}}{\ensuremath{\overleftarrow{E}_{#1}}}
\DeclareDocumentCommand{\VESigmaD}{O{\req} O{\sigma}}{\ensuremath{E_{#1,#2}}}

\DeclareDocumentCommand{\NodeVRange}{O{\req} O{i} O{j}}{\ensuremath{V^N_{#1,#2,#3}}}
\DeclareDocumentCommand{\NodeVRangeRange}{O{\req} O{i} O{j} O{u} O{v}}{\ensuremath{V^N_{#1,#2,#3,#4,#5}}}

\DeclareDocumentCommand{\path}{O{k }}{\ensuremath{P_{#1}}}
\DeclareDocumentCommand{\cycle}{O{k }}{\ensuremath{C_{#1}}}

\DeclareDocumentCommand{\NodePaths}{O{\req}}{\ensuremath{\mathcal{P}_{#1}}}

\DeclareDocumentCommand{\loadV}{O{\req} O{u}}{\ensuremath{l_{#1,#2}}}
\DeclareDocumentCommand{\loadE}{O{\req} O{u} O{v}}{\ensuremath{l_{#1,#2,#3}}}
\DeclareDocumentCommand{\loadX}{O{\req} O{x}}{\ensuremath{l_{#1,#2}}}
\DeclareDocumentCommand{\decomp}{O{\req} O{k}}{\ensuremath{D_{#1}^{#2}}}
\DeclareDocumentCommand{\decompHat}{O{\req} O{k}}{\ensuremath{{\hat{D}}_{#1}^{#2}}}
\DeclareDocumentCommand{\load}{O{\req} O{k}}{\ensuremath{l_{#1}^{#2}}}
\DeclareDocumentCommand{\prob}{O{\req} O{k}}{\ensuremath{f_{#1}^{#2}}}
\DeclareDocumentCommand{\mapping}{O{\req} O{k}}{\ensuremath{m_{#1}^{#2}}}

\DeclareDocumentCommand{\loadHat}{O{\req} O{k}}{\ensuremath{\hat{l}_{#1}^{#2}}}
\DeclareDocumentCommand{\probHat}{O{\req} O{k}}{\ensuremath{\hat{f}_{#1}^{#2}}}
\DeclareDocumentCommand{\mappingHat}{O{\req} O{k}}{\ensuremath{\hat{m}_{#1}^{#2}}}

\DeclareDocumentCommand{\loadHat}{O{\req} O{k}}{\ensuremath{\hat{l}_{#1}^{#2}}}
\DeclareDocumentCommand{\probHat}{O{\req} O{k}}{\ensuremath{\hat{f}_{#1}^{#2}}}
\DeclareDocumentCommand{\mappingHat}{O{\req} O{k}}{\ensuremath{\hat{m}_{#1}^{#2}}}

\DeclareDocumentCommand{\randVarX}{O{\req} O{k}}{\ensuremath{X_{#1}^{#2}}}
\DeclareDocumentCommand{\randVarY}{O{\req}}{\ensuremath{Y_{#1}}}
\DeclareDocumentCommand{\randVarZ}{O{\req}}{\ensuremath{Z_{#1}}}
\DeclareDocumentCommand{\randVarL}{O{x}}{\ensuremath{L_{x}}}
\DeclareDocumentCommand{\randVarLX}{O{\req} O{x} O{y}}{\ensuremath{L_{#1,#2,#3}}}
\DeclareDocumentCommand{\randVarLNode}{O{\req} O{\type} O{u}}{\ensuremath{L_{#1,#2,#3}}}
\DeclareDocumentCommand{\randVarLEdge}{O{\req} O{u} O{v}}{\ensuremath{L_{#1,#2,#3}}}
\DeclareDocumentCommand{\randVarM}{O{\req}}{\ensuremath{M_{#1}}}
\DeclareDocumentCommand{\randVarC}{O{\req}}{\ensuremath{C_{#1}}}

\DeclareDocumentCommand{\ProbVarX}{O{1}}{\ensuremath{\mathbb{P}(\randVarX = #1)}}
\DeclareDocumentCommand{\ProbVarY}{O{1}}{\ensuremath{\mathbb{P}(\randVarY = #1)}}
\DeclareDocumentCommand{\ProbVarZ}{O{1}}{\ensuremath{\mathbb{P}(\randVarZ = #1)}}
\DeclareDocumentCommand{\ProbVarL}{O{1}}{\ensuremath{\mathbb{P}(\randVarL = #1)}}
\DeclareDocumentCommand{\ProbVarM}{O{1}}{\ensuremath{\mathbb{P}(\randVarM = #1)}}
\DeclareDocumentCommand{\ProbVarC}{O{1}}{\ensuremath{\mathbb{P}(\randVarC = #1)}}

\DeclareDocumentCommand{\WAC}{O{\req}}{\ensuremath{\textnormal{WC}_{\req}}}

\DeclareDocumentCommand{\PotEmbeddings}{O{\req}}{\ensuremath{\mathcal{D}_{#1}}}
\DeclareDocumentCommand{\PotEmbeddingsHat}{O{\req}}{\ensuremath{\hat{\mathcal{D}}_{#1}}}

\DeclareDocumentCommand{\maxLoadX}{O{x}}{\ensuremath{\textnormal{max}^{L,\sum}_{#1}}}
\DeclareDocumentCommand{\maxLoadV}{O{\req} O{\type} O{u}}{\ensuremath{\textnormal{max}^{L}_{#1,#2,#3}}}
\DeclareDocumentCommand{\maxLoadE}{O{\req} O{u} O{v}}{\ensuremath{\textnormal{max}^{L}_{#1,#2,#3}}}
\DeclareDocumentCommand{\maxLoadVSum}{O{\req} O{\type} O{u}}{\ensuremath{\textnormal{max}^{L,\Sigma}_{#1,#2,#3}}}
\DeclareDocumentCommand{\maxLoadESum}{O{u} O{v}}{\ensuremath{\textnormal{max}^{L,\Sigma}_{#1,#2}}}

\DeclareDocumentCommand{\VVroot}{O{\req}}{\ensuremath{s_{#1}}}
\DeclareDocumentCommand{\VVpred}{O{\req}}{\ensuremath{\pi_{#1}}}

\newcommand{\extractionOrderCharacter}{\ensuremath{\mathcal{X}}}

\DeclareDocumentCommand{\Cycles}{O{\req}}{\ensuremath{\mathcal{C}_{#1}}}
\DeclareDocumentCommand{\Paths}{O{\req}}{\ensuremath{\mathcal{P}_{#1}}}

\DeclareDocumentCommand{\VVcycleSource}{O{\req} O{k}}{\ensuremath{s^{C_{#2}}_{#1}}}
\DeclareDocumentCommand{\VVcycleTarget}{O{\req} O{k}}{\ensuremath{t^{C_{#2}}_{#1}}}
\DeclareDocumentCommand{\VVpathSource}{O{\req} O{k}}{\ensuremath{s^{P_{#2}}_{#1}}}
\DeclareDocumentCommand{\VVpathTarget}{O{\req} O{k}}{\ensuremath{t^{P_{#2}}_{#1}}}

\DeclareDocumentCommand{\VGcycle}{O{\req} O{k}}{\ensuremath{G^{\extractionOrderCharacter,C_{#2}}_{#1}}}
\DeclareDocumentCommand{\VVcycle}{O{\req} O{k}}{\ensuremath{V^{\extractionOrderCharacter,C_{#2}}_{#1}}}
\DeclareDocumentCommand{\VEcycle}{O{\req} O{k}}{\ensuremath{E^{\extractionOrderCharacter,C_{#2}}_{#1}}}
\DeclareDocumentCommand{\VEcycleSame}{O{\req} O{k}}{\ensuremath{\overrightarrow{E}^{C_{#2}}_{#1}}}
\DeclareDocumentCommand{\VEcycleDiff}{O{\req} O{k}}{\ensuremath{\overleftarrow{E}^{C_{#2}}_{#1}}}

\DeclareDocumentCommand{\VGcycleOrig}{O{\req} O{k}}{\ensuremath{G^{C_{#2}}_{#1}}}
\DeclareDocumentCommand{\VVcycleOrig}{O{\req} O{k}}{\ensuremath{V^{C_{#2}}_{#1}}}
\DeclareDocumentCommand{\VEcycleOrig}{O{\req} O{k}}{\ensuremath{E^{C_{#2}}_{#1}}}
\DeclareDocumentCommand{\VGpath}{O{\req} O{k}}{\ensuremath{G^{P_{#2}}_{#1}}}
\DeclareDocumentCommand{\VVpath}{O{\req} O{k}}{\ensuremath{V^{P_{#2}}_{#1}}}
\DeclareDocumentCommand{\VEpath}{O{\req} O{k}}{\ensuremath{E^{P_{#2}}_{#1}}}

\DeclareDocumentCommand{\VEpathSame}{O{\req} O{k}}{\ensuremath{\overrightarrow{E}^{P_{#2}}_{#1}}}
\DeclareDocumentCommand{\VEpathDiff}{O{\req} O{k}}{\ensuremath{\overleftarrow{E}^{P_{#2}}_{#1}}}

\DeclareDocumentCommand{\VEDiff}{O{\req}}{\ensuremath{\overleftarrow{E}^{\extractionOrderCharacter}_{#1}}}

\DeclareDocumentCommand{\VEcycles}{O{\req}}{\ensuremath{E^{\mathcal{C}}_{#1}}}
\DeclareDocumentCommand{\VEpaths}{O{\req}}{\ensuremath{E^{\mathcal{P}}_{#1}}}

\DeclareDocumentCommand{\VVcycleSourcesTargets}{O{\req}}{\ensuremath{V^{\mathcal{C},\pm}_{#1}}}
\DeclareDocumentCommand{\VVpathSourcesTargets}{O{\req}}{\ensuremath{V^{\mathcal{P},\pm}_{#1}}}
\DeclareDocumentCommand{\VVSourcesTargets}{O{\req}}{\ensuremath{V^{\pm}_{#1}}}

\DeclareDocumentCommand{\VVcycleSources}{O{\req}}{\ensuremath{V^{\mathcal{C},+}_{#1}}}
\DeclareDocumentCommand{\VVpathSources}{O{\req}}{\ensuremath{V^{\mathcal{P},+}_{#1}}}

\DeclareDocumentCommand{\VVcycleTargets}{O{\req}}{\ensuremath{V^{\mathcal{C},-}_{#1}}}
\DeclareDocumentCommand{\VVpathTargets}{O{\req}}{\ensuremath{V^{\mathcal{P},-}_{#1}}}

\DeclareDocumentCommand{\VGcycleBranchR}{O{\req} O{k}}{\ensuremath{G^{C_{#2}, B_1}_{#1}}}
\DeclareDocumentCommand{\VVcycleBranchR}{O{\req} O{k}}{\ensuremath{V^{C_{#2}, B_1}_{#1}}}
\DeclareDocumentCommand{\VEcycleBranchR}{O{\req} O{k}}{\ensuremath{E^{C_{#2}, B_1}_{#1}}}

\DeclareDocumentCommand{\VGcycleBranchL}{O{\req} O{k}}{\ensuremath{G^{C_{#2}, B_2}_{#1}}}
\DeclareDocumentCommand{\VVcycleBranchL}{O{\req} O{k}}{\ensuremath{V^{C_{#2}, B_2}_{#1}}}
\DeclareDocumentCommand{\VEcycleBranchL}{O{\req} O{k}}{\ensuremath{E^{C_{#2}, B_2}_{#1}}}

\DeclareDocumentCommand{\VGdecomp}{O{\req} O{k}}{\ensuremath{G^{\mathcal{D}}_{#1}}}
\DeclareDocumentCommand{\VVdecomp}{O{\req} O{k}}{\ensuremath{V^{\mathcal{D}}_{#1}}}
\DeclareDocumentCommand{\VEdecomp}{O{\req} O{k}}{\ensuremath{E^{\mathcal{D}}_{#1}}}

\DeclareDocumentCommand{\VVbranching}{O{\req} }{\ensuremath{\mathcal{B}_{#1}}}
\DeclareDocumentCommand{\VVbranchingcycle}{O{\req} O{k}}{\ensuremath{\mathcal{B}^{C_{#2}}_{#1}}}
\DeclareDocumentCommand{\VVbranchingpath}{O{\req} O{k}}{\ensuremath{\mathcal{B}^{P_{k}}_{#1}}}
\DeclareDocumentCommand{\VVjoin}{O{\req} }{\ensuremath{\mathcal{J}_{#1}}}
\DeclareDocumentCommand{\VVaggregation}{O{\req} }{\ensuremath{\mathcal{A}_{#1}}}

\DeclareDocumentCommand{\VGextcycle}{O{\req} O{k}}{\ensuremath{G^{C_{#2}}_{#1,\textnormal{ext}}}}

\DeclareDocumentCommand{\VVextcycle}{O{\req} O{k}}{\ensuremath{V^{C_{#2}}_{#1,\textnormal{ext}}}}
\DeclareDocumentCommand{\VVextcycleSources}{O{\req} O{k}}{\ensuremath{V^{C_{#2}}_{#1,+}}}
\DeclareDocumentCommand{\VVextcycleTargets}{O{\req} O{k}}{\ensuremath{V^{C_{#2}}_{#1,-}}}
\DeclareDocumentCommand{\VVextcycleSubstrate}{O{\req} O{k}}{\ensuremath{V^{C_{#2}}_{#1,S}}}

\DeclareDocumentCommand{\VEextcycle}{O{\req} O{k}}{\ensuremath{E^{C_{#2}}_{#1,\textnormal{ext}}}}
\DeclareDocumentCommand{\VEextcycleSources}{O{\req} O{k}}{\ensuremath{E^{C_{#2}}_{#1,+}}}
\DeclareDocumentCommand{\VEextcycleTargets}{O{\req} O{k}}{\ensuremath{E^{C_{#2}}_{#1,-}}}
\DeclareDocumentCommand{\VEextcycleSubstrate}{O{\req} O{k}}{\ensuremath{E^{C_{#2}}_{#1,S}}}
\DeclareDocumentCommand{\VEextcycleF}{O{\req} O{k}}{\ensuremath{E^{C_{#2}}_{#1,F}}}

\DeclareDocumentCommand{\VGextpath}{O{\req} O{k}}{\ensuremath{G^{{P_{#2}}}_{#1,\textnormal{ext}}}}

\DeclareDocumentCommand{\VVextpath}{O{\req} O{k}}{\ensuremath{V^{P_{#2}}_{#1,\textnormal{ext}}}}
\DeclareDocumentCommand{\VVextpathSources}{O{\req} O{k}}{\ensuremath{V^{P_{#2}}_{#1,+}}}
\DeclareDocumentCommand{\VVextpathTargets}{O{\req} O{k}}{\ensuremath{V^{P_{#2}}_{#1,-}}}
\DeclareDocumentCommand{\VVextpathSubstrate}{O{\req} O{k}}{\ensuremath{V^{P_{#2}}_{#1,S}}}

\DeclareDocumentCommand{\VEextpath}{O{\req} O{k}}{\ensuremath{E^{P_{#2}}_{#1,\textnormal{ext}}}}
\DeclareDocumentCommand{\VEextpathSources}{O{\req} O{k}}{\ensuremath{E^{P_{#2}}_{#1,+}}}
\DeclareDocumentCommand{\VEextpathTargets}{O{\req} O{k}}{\ensuremath{E^{P_{#2}}_{#1,-}}}
\DeclareDocumentCommand{\VEextpathSubstrate}{O{\req} O{k}}{\ensuremath{E^{P_{#2}}_{#1,S}}}
\DeclareDocumentCommand{\VEextpathF}{O{\req} O{k}}{\ensuremath{E^{P_{#2}}_{#1,F}}}

\DeclareDocumentCommand{\forest}{O{\req}}{\ensuremath{\mathcal{F}_{#1}}}

\DeclareDocumentCommand{\VGforest}{O{\req}}{\ensuremath{G^{\mathcal{A},\mathcal{F}}_{#1}}}
\DeclareDocumentCommand{\VVforest}{O{\req}}{\ensuremath{V^{\mathcal{A},\mathcal{F}}_{#1}}}
\DeclareDocumentCommand{\VEforest}{O{\req}}{\ensuremath{E^{\mathcal{A},\mathcal{F}}_{#1}}}

\DeclareDocumentCommand{\VGforestOrig}{O{\req}}{\ensuremath{G^{\mathcal{F}}_{#1}}}
\DeclareDocumentCommand{\VVforestOrig}{O{\req}}{\ensuremath{V^{\mathcal{F}}_{#1}}}
\DeclareDocumentCommand{\VEforestOrig}{O{\req}}{\ensuremath{E^{\mathcal{F}}_{#1}}}

\DeclareDocumentCommand{\varFlowInput}{O{\req} O{i} O{u}}{\ensuremath{f^+_{#1,#2,#3}}}
\DeclareDocumentCommand{\varFlowOutput}{O{\req} O{i} O{u}}{\ensuremath{f^+_{#1,#2,#3}}}

\DeclareDocumentCommand{\VEextcycleHorizontal}{O{\req} O{k} O{u} O{v}}{\ensuremath{E^{C_{#2}}_{#1,\textnormal{ext},#3,#4}}}
\DeclareDocumentCommand{\VEextpathHorizontal}{O{\req} O{k} O{u} O{v}}{\ensuremath{E^{P_{#2}}_{#1,\textnormal{ext},#3,#4}}}
\DeclareDocumentCommand{\VEextcycleVertical}{O{\req} O{k} O{\type} O{u}}{\ensuremath{E^{C_{#2}}_{#1,\textnormal{ext},#3,#4}}}
\DeclareDocumentCommand{\VEextpathVertical}{O{\req} O{k} O{\type} O{u}}{\ensuremath{E^{P_{#2}}_{#1,\textnormal{ext},#3,#4}}}

\DeclareDocumentCommand{\VEextCGHorizontal}{O{\req} O{u} O{v}}{\ensuremath{E^{\textnormal{ext,SCG}}_{#1,#2,#3}}}
\DeclareDocumentCommand{\VEextCGVertical}{O{\req} O{\type} O{u}}{\ensuremath{E^{\textnormal{ext,SCG}}_{#1,#2,#3}}}

\DeclareDocumentCommand{\VVextCGFlowNodes}{O{\req}}{\ensuremath{V^{\textnormal{ext,SCG}}_{#1,\textnormal{flow}}}}
\DeclareDocumentCommand{\VEextCGFlowEdges}{O{\req}}{\ensuremath{E^{\textnormal{ext,SCG}}_{#1,\textnormal{flow}}}}

\DeclareDocumentCommand{\VGextcycleFlow}{O{\req} O{k}}{\ensuremath{G^{C_{#2}}_{#1,\textnormal{ext},f}}}

\DeclareDocumentCommand{\VGextcycleFlowBranchR}{O{\req} O{k}}{\ensuremath{G^{C_{#2},{B}_1}_{#1,\textnormal{ext},f}}}
\DeclareDocumentCommand{\VVextcycleFlowBranchR}{O{\req} O{k}}{\ensuremath{V^{C_{#2},{B}_1}_{#1,\textnormal{ext},f}}}
\DeclareDocumentCommand{\VEextcycleFlowBranchR}{O{\req} O{k}}{\ensuremath{E^{C_{#2},{B}_1}_{#1,\textnormal{ext},f}}}

\DeclareDocumentCommand{\VGextcycleFlowBranchL}{O{\req} O{k}}{\ensuremath{G^{C_{#2},{B}_2}_{#1,\textnormal{ext},f}}}
\DeclareDocumentCommand{\VVextcycleFlowBranchL}{O{\req} O{k}}{\ensuremath{V^{C_{#2},{B}_2}_{#1,\textnormal{ext},f}}}
\DeclareDocumentCommand{\VEextcycleFlowBranchL}{O{\req} O{k}}{\ensuremath{E^{C_{#2},{B}_2}_{#1,\textnormal{ext},f}}}

\DeclareDocumentCommand{\VGextcycleBranchR}{O{\req} O{k}}{\ensuremath{G^{C_{#2},{B}_1}_{#1,\textnormal{ext}}}}
\DeclareDocumentCommand{\VVextcycleBranchR}{O{\req} O{k}}{\ensuremath{V^{C_{#2},{B}_1}_{#1,\textnormal{ext}}}}
\DeclareDocumentCommand{\VEextcycleBranchR}{O{\req} O{k}}{\ensuremath{E^{C_{#2},{B}_1}_{#1,\textnormal{ext}}}}

\DeclareDocumentCommand{\VGextcycleBranchL}{O{\req} O{k}}{\ensuremath{G^{C_{#2},{B}_2}_{#1,\textnormal{ext}}}}
\DeclareDocumentCommand{\VVextcycleBranchL}{O{\req} O{k}}{\ensuremath{V^{C_{#2},{B}_2}_{#1,\textnormal{ext}}}}
\DeclareDocumentCommand{\VEextcycleBranchL}{O{\req} O{k}}{\ensuremath{E^{C_{#2},{B}_2}_{#1,\textnormal{ext}}}}

\DeclareDocumentCommand{\VGextpathFlow}{O{\req} O{k}}{\ensuremath{G^{P_{#2}}_{#1,\textnormal{ext},f}}}
\DeclareDocumentCommand{\VVextpathFlow}{O{\req} O{k}}{\ensuremath{V^{P_{#2}}_{#1,\textnormal{ext},f}}}
\DeclareDocumentCommand{\VEextpathFlow}{O{\req} O{k}}{\ensuremath{E^{P_{#2}}_{#1,\textnormal{ext},f}}}

\DeclareDocumentCommand{\VVKSource}{O{\req} O{K}}{\ensuremath{s^{K}_{#1}}}
\DeclareDocumentCommand{\VVKTarget}{O{\req} O{K}}{\ensuremath{t^{K}_{#1}}}
\DeclareDocumentCommand{\VVKSourcesTargets}{O{\req}}{\ensuremath{V^{K,\pm}_{#1}}}

\DeclareDocumentCommand{\VEbfsPre}{O{j} O{\req}}{\ensuremath{E^{\extractionOrderCharacter,\mathrm{pre}}_{#2,#1}}}
\DeclareDocumentCommand{\VEbfsSuc}{O{i} O{\req}}{\ensuremath{E^{\extractionOrderCharacter,\mathrm{suc}}_{#2,#1}}}

\DeclareDocumentCommand{\VEbfsInter}{O{i} O{j} O{\req}}{\ensuremath{E^{\extractionOrderCharacter}_{#3,#1\leadsto #2}}}

\DeclareDocumentCommand{\VEbfsLabels}{O{e} O{\req} }{\ensuremath{\mathcal{L}^{\extractionOrderCharacter}_{#2,#1}}}

\DeclareDocumentCommand{\VEbfsLabelsOrig}{O{e} O{\req} }{\ensuremath{\mathcal{L}_{#2,#1}}}

\DeclareDocumentCommand{\VEbfsAC}{O{i} O{j} O{\req}}{\ensuremath{C^{\extractionOrderCharacter}_{#1,#2}}}

\DeclareDocumentCommand{\VEbfsACL}{O{i} O{j} O{\req} }{\ensuremath{P^{1}_{#1,#2}}}
\DeclareDocumentCommand{\VEbfsACR}{O{i} O{j} O{\req} }{\ensuremath{P^{2}_{#1,#2}}}

\DeclareDocumentCommand{\VEbfsBags}{O{i} O{\req} }{\ensuremath{\mathcal{B}^{+}_{#2,#1}}}
\DeclareDocumentCommand{\VEbfsBagIterator}{}{\ensuremath{b}}
\DeclareDocumentCommand{\VEbfsBag}{O{\VEbfsBagIterator}}{\ensuremath{B^{\extractionOrderCharacter}_{#1}}}

\DeclareDocumentCommand{\deltaMinusA}{O{i}}{\ensuremath{\delta^-_{\extractionOrderCharacter}(#1)}}
\DeclareDocumentCommand{\deltaPlusA}{O{i}}{\ensuremath{\delta^+_{\extractionOrderCharacter}(#1)}}

\DeclareDocumentCommand{\deltaMinusA}{O{i}}{\ensuremath{\delta^-_{\extractionOrderCharacter}(#1)}}
\DeclareDocumentCommand{\deltaPlusA}{O{i}}{\ensuremath{\delta^+_{\extractionOrderCharacter}(#1)}}

\newcommand{\varStyle}[1]{\textnormal{\textrm{\textbf{#1}}}}

\makeatletter
\def\RemoveSpaces#1{\zap@space#1 \@empty}
\makeatother

\DeclareDocumentCommand{\var}{u{ } u{ }}{\text{$\langle$\,\varStyle{\RemoveSpaces{#1}}\,\ensuremath{|}\varStyle{\,\RemoveSpaces{#2}\,}$\rangle$}}

\newcommand{\compP}{\ensuremath{\mathcal{P}}}
\newcommand{\compNP}{\ensuremath{\mathcal{NP}}}
\newcommand{\compPeqNP}{\ensuremath{\compP{\,=\,}\compNP}}

\newcommand{\NPhard}{\ensuremath{\compNP\text{-hard}}}
\newcommand{\NPcomplete}{\ensuremath{\compNP\text{-complete}}}
\newcommand{\NPhardness}{\ensuremath{\compNP\text{-hardness}}}
\newcommand{\NPcompleteness}{\ensuremath{\compNP\text{-completeness}}}

\newcommand{\VNEP}{\ensuremath{\textsc{VNEP}}}
\newcommand{\DirEDPwC}{\ensuremath{\textsc{DirEDPwC}}}
\newcommand{\ThreeSAT}{\ensuremath{\text{3-}\textsc{SAT}}}
\newcommand{\SpecialThreeSAT}{\ensuremath{\textsc{4P3C\text{-3-}SAT}}}

\begin{document}
\IEEEoverridecommandlockouts
\IEEEpubid{\makebox[\columnwidth]{A shorter version of this paper will be presented at IFIP Networking 2018.~\hfill}\hspace{\columnsep}\makebox[\columnwidth]{}}

\title{\huge $\mathcal{NP}$-Completeness and Inapproximability of the \\Virtual Network Embedding Problem and Its Variants}

\author{\IEEEauthorblockN{Matthias Rost}
\IEEEauthorblockA{Technische Universit\"at Berlin\\
Email: mrost@inet.tu-berlin.de}
\and
\IEEEauthorblockN{Stefan Schmid}
\IEEEauthorblockA{University of Vienna\\
Email: stefan\_schmid@univie.ac.at}
}

\maketitle

\begin{abstract}
Many resource allocation problems in the cloud can
be described as a basic Virtual Network Embedding Problem ($\VNEP$):
the problem of finding a mapping of a 
\emph{request graph} (describing a workload)
onto a \emph{substrate graph} (describing the 
physical infrastructure). Applications range from
mapping testbeds (from where the problem originated), 
over the embedding of batch-processing workloads (virtual clusters)
to the embedding of service function chains. 
The different applications come with their own specific requirements and constraints,
including node mapping constraints, routing policies,
and latency constraints.
While the $\VNEP$ has been studied intensively over the last years, complexity results
are only known for specific models and we lack a comprehensive understanding of its hardness. 

This paper charts the complexity landscape of the $\VNEP$ by providing
a systematic analysis of the hardness of a wide range of VNEP variants, 
using a unifying and rigorous proof framework. 
In particular, we show that the problem of finding a feasible embedding is $\NPcomplete$ in general, and, hence, the $\VNEP$ cannot be approximated \emph{under any objective}, unless $\compPeqNP$ holds. Importantly, we derive $\NPcompleteness$ results also for finding approximate embeddings, which may violate, e.g., capacity constraints by certain factors. Lastly, we prove that our results still pertain when restricting the request graphs to planar or degree-bounded graphs.
\end{abstract}

\section{Introduction}

At the heart of the cloud computing paradigm lies the idea of 
efficient resource sharing: due to virtualization, 
multiple workloads can
co-habit and use a given resource infrastructure simultaneously.
Indeed, cloud computing introduces great flexibilities in terms of
\emph{where} workloads can be mapped. At the same time, exploiting this 
mapping flexibility poses a fundamental algorithmic challenge.
In particular, in order to provide predictable performance, guarantees on all, i.e. 
node and edge, resources need to be ensured. Indeed, it has been shown that cloud application performance can suffer significantly from interference on the communication network~\cite{talkabout}.

The underlying algorithmic problem is essentially a graph theoretical one:
both the workload as well as the infrastructure can be modeled as \emph{graphs}.
The former, the so-called \emph{request graph}, describes the resource requirements
both on the nodes (e.g., the virtual machines) as well as on the interconnecting network.
The latter, the so-called \emph{substrate graph}, describes the physical infrastructure
and its resources (servers and links). Figure~\ref{fig:vnep-example} depicts an example of embedding a request graph.

The problem is known in the networking community under the name
\emph{Virtual Network Embedding Problem} ($\VNEP$) and has been studied
intensively for over a decade~\cite{vnep,vnep-survey}. Besides the rather general study of the $\VNEP$, which emerged originally from the study of testbed provisioning, 
essentially the same problems are considered in the context of Service Function Chaining~\cite{mehraghdam2014specifying,rfc7665}, as well as in the context of embeddings Virtual Clusters, a specific batch processing request abstraction~\cite{ballani2011towards,ccr15emb}.

\begin{figure}[tb!]
\centering
\includegraphics[width=0.9\columnwidth]{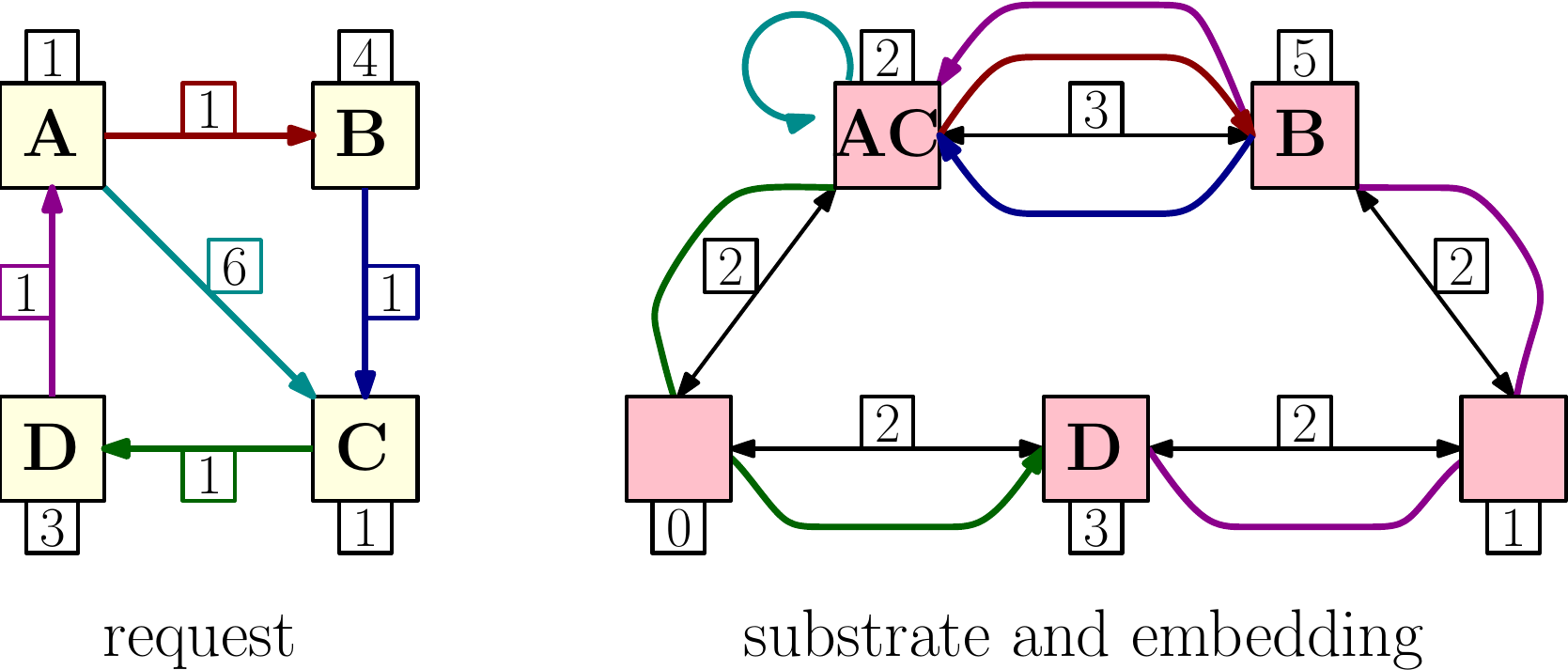}
\caption{Example of embedding a request (left) on a substrate network (right). The numeric labels at the network elements denote the demanded resources and the available capacity, respectively. In the embedding, request nodes \textbf{A} and \textbf{C} are collocated on the same substrate node, such that the request edge $(\textnormal{\textbf{A}},\textnormal{\textbf{C}})$ does not use any substrate edge resources.}

\label{fig:vnep-example}
\end{figure}

\begin{table*}[bth]

\setlength{\tabcolsep}{0.1em}
\newcommand\RotHeader[1]{\rotatebox{90}{\parbox{2.4cm}{\centering#1}}}

\caption{Overview on results obtained in this paper.}
\label{tab:results-overview-real}
\newcolumntype{P}[1]{>{\centering\arraybackslash}p{#1}}
\newcolumntype{L}[1]{>{\arraybackslash}p{#1}}
\newcolumntype{R}[1]{>{\raggedleft\arraybackslash}p{#1}}
\renewcommand{\arraystretch}{1.2}
\begin{tabular}{ P{0.3cm}| L{1.25cm} | R{6cm} P{0.35cm} | R{4.6cm} | c c c c c | }
\multicolumn{1}{L{0.3cm}}{~} & \multicolumn{1}{L{1.25cm}}{~} & \multicolumn{1}{R{6cm}}{~} & \multicolumn{1}{P{0.35cm}}{~} & \multicolumn{1}{R{4.6cm}}{~} & \multicolumn{1}{c}{~} & \multicolumn{1}{c}{~} & \multicolumn{1}{c}{~} & \multicolumn{1}{c}{~} & \multicolumn{1}{c}{~} \\[-14pt]
\cline{5-10}
\multicolumn{1}{c}{} & \multicolumn{1}{c}{} & \multicolumn{1}{c }{} & \multirow{6}{*}[0.1em]{
\RotHeader{\textbf{$\VNEP$ variants}}
}
 & Identifier according to Definition~\ref{def:nomenclature} 
&  $\var VE - $ & $\var E N $ & $\var V R $ & $\var - NR $ & $\var - NL $  \\
\cline{5-10}
\multicolumn{1}{c}{} & \multicolumn{1}{c}{} & \multicolumn{1}{c}{} &  \multicolumn{1}{ c| }{} & Enforcing Node Capacities & \checkmark & $\star$ & \checkmark & $\star$ & $\star$ \\ \cline{5-10}
\multicolumn{1}{c}{} & \multicolumn{1}{c}{} & \multicolumn{1}{c}{} & \multicolumn{1}{ c |}{} & Enforcing Edge Capacities & \checkmark & \checkmark & $\star$ & $\star$ &  $\star$\\ \cline{5-10}
\multicolumn{1}{c}{} & \multicolumn{1}{c}{} & \multicolumn{1}{c}{} & \multicolumn{1}{ c |}{}  & Enforcing Node Placement Restrictions & $\star$ & \checkmark & $\star$ & \checkmark & \checkmark \\ \cline{5-10}
\multicolumn{1}{c}{} & \multicolumn{1}{c}{} & \multicolumn{1}{c}{} & \multicolumn{1}{ c |}{} & Enforcing Edge Routing Restrictions & $\star$&$\star$ & \checkmark & \checkmark & $\star$\\ \cline{5-10}
\multicolumn{1}{c}{} & \multicolumn{1}{c}{} & \multicolumn{1}{c}{} & \multicolumn{1}{ c |}{}   & Enforcing Latency Restrictions & $\star$&$\star$ & $\star$& $\star$& \checkmark \\
\cline{5-10}
\multicolumn{8}{c}{~} \\[-8pt] 
\cline{2-10}

\multirow{6}{*}{\RotHeader{\textbf{Results}}} & 
 Section~\ref{sec:np-completeness-of-the-vnep} &
\multicolumn{3}{ R{10.9cm}| }{$\mathcal{NP}$-completeness and inapproximability under any objective}
& Thm.~\ref{thm:vnep-capacity-nodes-and-edges-is-np-complete} & Thm.~\ref{thm:decision-edge-capacities-node-placement} & Thm.~\ref{thm:node-capacities-routing-restrictions} & Thm.~\ref{thm:without-capacities-node-placement-and-routing-or-latencies} & Thm.~\ref{thm:without-capacities-node-placement-and-routing-or-latencies}\\
\cline{2-10} 

& \multirow{3}{*}{\raisebox{-0.17cm}{Section~\ref{sec:np-completeness-approximate-embeddings}}}  &
\multicolumn{3}{ R{10.9cm} | }{$\NPcompleteness$  and inapproximability when increasing \emph{node capacities} by a factor $\alpha < 2$}  & Thm.~\ref{thm:alpha-approximate-embeddings} & - & Thm.~\ref{thm:alpha-approximate-embeddings} & - & - \\
\cline{3-10}

& & \multicolumn{3}{ R{10.9cm}| }{Inapproximability when increasing edge capacities by a factor $\beta \in \Theta(\log n / \log \log n)$  (unless $\compNP \subseteq \mathcal{BP}\text{-}\textit{TIME}(\bigcup_{d \geq 1} n^{d \log \log n})$)} & \raisebox{-0.17cm}{Thm.~\ref{thm:beta-inapproximability}} & \raisebox{-0.17cm}{Thm.~\ref{thm:beta-inapproximability}} & \raisebox{-0.17cm}{-}& \raisebox{-0.17cm}{-}& \raisebox{-0.17cm}{-} \\
\cline{3-10}

& & \multicolumn{3}{ R{10.9cm}| }{$\NPcompleteness$ and inapproximability when loosening latency bounds by a factor $\gamma < 2$}  & - & - & - & -  &  Thm.~\ref{thm:gamma-approximate-embeddings} \\
\cline{2-10} 

& \multirow{2}{*}{
\hspace{-5pt} Section~\ref{sec:np-completeness-graph-restrictions}} &
\multicolumn{3}{ R{10.9cm}| }{Results are preserved for acyclic substrates (except for Thm.~\ref{thm:beta-inapproximability})}& \multicolumn{5}{c| }{\raisebox{0.05cm}{\rule{1.8cm}{0.4pt}}~Obs.~\ref{obs:vnep-on-dags}~ \raisebox{0.05cm}{\rule{1.8cm}{0.4pt}}}\\
\cline{3-10}
& & \multicolumn{3}{ R{10.9cm}| }{Results are preserved for acyclic, planar, degree-bounded requests} & \multicolumn{5}{c| }{\raisebox{0.05cm}{\rule{1.8cm}{0.4pt}}~Thm.~\ref{thm:np-completeness-on-restricted-request-graphs}~\raisebox{0.05cm}{\rule{1.8cm}{0.4pt}}}	 \\
\cline{2-10}
\end{tabular}
\end{table*}

\subsection{Related Work}
\label{sec:related-work}

\paragraph{Objectives \& Restrictions}

Depending on the setting, many different objectives are considered
for the $\VNEP$. The most studied ones concern minimizing the (resource allocation) cost~\cite{vnep, vnep-survey}, maximizing the profit by exerting admission control~\cite{even2013competitive,rostSchmidFeldmann2014}, and minimizing the maximal load~\cite{mehraghdam2014specifying,chowdhury2012vineyard}.

Besides commonly enforcing that the substrate's physical capacities on servers and edges are not exceeded to provide Quality-of-Service~\cite{vnep-survey}, additional restrictions have emerged:
\begin{itemize}
\item Restrictions on the placement of virtual nodes first arose to enforce closeness to locations of interest~\cite{vnep}, but were also used in the context of privacy policies to restrict mappings to certain countries~\cite{schaffrath2012optimizing}. However, these restrictions are now also used in the context of Service Function Chaining, as specific functions may only be mapped on x86 servers, while firewall appliances cannot~\cite{mehraghdam2014specifying, rfc7665}.
\item Routing restrictions first arose in the context of expressing security policies, as for example some traffic may not be routed via insecure domains or physical links shall not be shared with other virtual networks~\cite{vnep-survey,bays2012security}.
\item Restrictions on latencies were studied for the  $\VNEP$ in~\cite{infuhr2011introducing} and have been recently studied intensely in the context of Service Function Chaining to achieve responsiveness and Quality-of-Service~\cite{mehraghdam2014specifying,rfc7665}.
\end{itemize}

\paragraph{Algorithmic Approaches} 

Several dozens of algorithms were proposed to solve the $\VNEP$ and its siblings, including the Virtual Cluster Embedding~\cite{ballani2011towards} and Service Function Chain Embedding problem~\cite{vnep-survey}. Most approaches to solve the VNEP either rely on heuristics~\cite{vnep} or metaheuristics~\cite{vnep-survey}. On the other hand, several works study exact (non-polynomial time) algorithms to solve the problem to \mbox{(near-)optimality} or to devise heuristics. Mixed Integer Programming is the most widely used exact approach~\cite{mehraghdam2014specifying, rostSchmidFeldmann2014, infuhr2011introducing}.

Only recently, approximation algorithms providing quality guarantees for the VNEP have been presented. In particular, the embedding of chains is approximated under assumptions on the requested resources and the achievable benefit in~\cite{sirocco16path}. In \cite{rostSchmidLeveragingRRIFIPwithPreprint} approximations for cactus request graphs are detailed, while \cite{rostSchmidFPTApproximations} presents fixed-parameter tractable approximations for \emph{arbitrary} request graph topologies.

\paragraph{Complexity Results}

Surprisingly, despite the relevance of the problem and the large body of literature, 
the complexity of the underlying problems has not received much attention. While it can be easily seen that the Virtual Network Embedding Problem encompasses several $\compNP$-hard problems as e.g. the $k$-disjoint paths problem~\cite{chuzhoy2007hardness}, the minimum linear arrangment problem~\cite{mla-survey}, 
or the subgraph isomorphism problem~\cite{eppstein1995subgraph}, most works on the $\VNEP$ cite a $\NPhardness$ result contained in a technical report from 2002 by Andersen~\cite{andersen2002theoretical}. The only other work studying the computational complexity is one by Amaldi et al.~\cite{amaldi2016computational}, which proved the $\NPhardness$ and inapproximability of the profit maximization objective while not taking into account latency or routing restrictions and not considering the hardness of embedding a \emph{single} request. 

\subsection{Contributions and Overview}

In this work, we initiate the systematic 
study of the computational complexity of the $\VNEP$. Taking all the aforementioned restrictions into account, we first compile a concise taxonomy of the $\VNEP$ variants in Section~\ref{sec:formal-model}. Then, we present a powerful reduction framework in Section~\ref{sec:reduction-framework}, which is the base for nearly all hardness results presented in this paper. In particular, we show the following (see also Table~\ref{tab:results-overview-real}):
\begin{itemize}
\item We show the $\NPcompleteness$ of \emph{five} different $\VNEP$ variants in Section~\ref{sec:np-completeness-of-the-vnep}. For example, we consider the variant only enforcing capacity constraints, but also one in which only node placement and latency restrictions must be obeyed \emph{in the absence of capacity constraints}.
\item We extend these results in Section~\ref{sec:np-completeness-approximate-embeddings} and show that the considered variants remain $\NPcomplete$ even when computing \emph{approximate} embeddings, which may exceed latency or capacity constraints by certain factors.
\item Lastly, we show in Section~\ref{sec:np-completeness-graph-restrictions} that the respective $\VNEP$ variants remain $\NPcomplete$ even when restricting substrate graphs to directed acyclic graphs (DAGs) and request graphs to planar, degree-bounded DAGs.
\end{itemize}
As we are proving $\NPcompleteness$ throughout this paper, the implications of our results are severe. Given the $\NPcompleteness$ of finding \emph{any} feasible solution, finding an optimal solution \emph{subject to any objective} is at least $\NPhard$. Furthermore, unless $\compPeqNP$ holds, the respective variants cannot be approximated to within \emph{any} factor. 

Table~\ref{tab:results-overview-real} summarizes our results and is to be read as follows. Any of the five rightmost columns represents a specific $\VNEP$ variant. The $\checkmark$ symbol indicates restrictions that are enforced, while the $\star$ symbol indicates restrictions which are not considered. Importantly, enabling a $\star$ restriction, does not change the results (cf. Lemma~\ref{lem:vnep-only-gets-harder}). Considering a specific variant, the respective column should be read from top to bottom. For example, for $\var VE - $, its $\NPcompleteness$ is shown in Theorem~\ref{thm:vnep-capacity-nodes-and-edges-is-np-complete} while its inapproximability when relaxing edge capacity constraints is shown in Theorem~\ref{thm:beta-inapproximability}. Lastly, all results also hold under the graph restrictions of the two bottom rows.

\section{Formal Model}
\label{sec:formal-model}

Within this section we formalize the $\VNEP$, introduce its variants, and lastly provide an Integer Programming formulation to solve any of the variants.

\paragraph*{Notation} The following notation is used throughout this work. We use $[x]$ to denote $\{1,2,\ldots,x\}$ for $x \in \mathbb{N}$. For a directed graph $G=(V,E)$, we denote by $\delta^+(v) \subseteq E$ and $\delta^-(v)\subseteq E$ the outgoing and incoming edges of node $v \in V$. When considering functions on tuples, we omit the parantheses of the tuple and simply write $f(a,b)$ instead of $f((a,b))$.

\subsection{Basic Problem Definition}
\newcommand{\simplePaths}{\ensuremath{\mathcal{P}_S}}

We refer to the physical network as substrate network and model it as directed graph $\SG=(\SV,\SE)$. Capacities in the substrate are given by the function $\Scap: \SV \cup \SE \to \mathbb{R}_{\geq 0} \cup \{\infty\}$. The capacity $\Scap(u)$ of node $u \in \SV$ may represent for example the number of CPUs while the capacity $\Scap(u,v)$ of edge $(u,v) \in \SE$ represents the available bandwidth. By allowing to set substrate capacities to $\infty$, the capacity constraints on the respective substrate elements can be effectively disabled. We denote by $\simplePaths$ the set of all simple paths in $\SG$.

A request is similarly modeled as directed graph \mbox{$\VG=(\VV,\VE)$} together with node and edge capacities (demands) \mbox{$\Vcap: \VV \cup \VE \to \mathbb{R}_{\geq 0}$}. 

The task is to find a \emph{mapping} of request graph $\VG$ on the substrate network $\SG$, i.e. a mapping of request nodes to substrate nodes and a mapping of request edges to paths in the substrate. Virtual nodes and edges can only be mapped on substrate nodes and edges of sufficient capacity. Accordingly, we denote by $\VVloc = \{u \in \SV | \Scap(u) \geq \Vcap(i)\}$ the set of substrate nodes supporting the mapping of node $i \in \VV$ and by $\VEloc = \{(u,v) \in \SE | \Scap(u,v) \geq \Vcap(i,j) \}$ the  substrate edges supporting the mapping of virtual edge $(i,j) \in \VE$.
\begin{definition}[Valid Mapping]
\label{def:valid-mapping}
A \emph{valid} mapping of request~$\VG$ to the substrate $\SG$ is a tuple~$\map=(\mapV, \mapE)$ of functions that map nodes and edges, respectively, s.t. the following holds:
\begin{itemize}
\item The function $\mapV: \VV \to \SV$ maps virtual nodes to \emph{suitable} substrate nodes, such that $\mapV(i) \in \VVloc$ holds for $i \in \VV$.
\item The function $\mapE: \VE \to \simplePaths$ maps virtual edges $(i,j) \in \VE$ to simple paths in $\SG$ connecting $\mapV(i)$ to $\mapV(j)$, such that $\mapE(i,j) \subseteq \VEloc$ holds for $(i,j) \in \VE$.
\end{itemize} 
\vspace{-10pt}
\end{definition}

Considering the above definition, note the following. Firstly, the mapping $\mapE(i,j)$ of the virtual edge $(i,j) \in \VE$ may be empty, if (and only if) $i$ and $j$ are mapped on the same substrate node.
Secondly, the definition only enforces that single resource allocations do not exceed the available capacity. To enforce that the cumulative allocations respect capacities, we introduce the following:

\begin{definition}[Allocations]
We denote by \mbox{$A_m(x) \in \mathbb{R}_{\geq 0}$} the resource allocations induced by valid mapping \mbox{$\map=(\mapV,\mapE)$} on substrate element $x \in \SG$ and define
\begin{alignat*}{5}
\huge
A_m(u) & = && \textstyle \sum_{ i \in \VV: \mapV(i)=u} \Vcap(i) \\
A_m(u,v) & = &&  \textstyle \sum_{(i,j)\in \VE: (u,v) \in \mapE(i,j)} \Vcap(i,j)
\end{alignat*}
for node $u \in \SV$ and edge $(u,v) \in \SE$, respectively.
\end{definition}

We call a mapping \emph{feasible}, if the (cumulative) allocations do not exceed the capacity of any substrate element:

\begin{definition}[Feasible Embedding]
\label{def:feasible-embedding}
A mapping $\map$ is a feasible embedding, if the allocations do not exceed the capacity, i.e. $A_m(x) \leq \Scap(x)$ holds for $x \in \SG$.
\end{definition}

In this paper we study the \emph{decision} variant of the $\VNEP$, asking whether there exists a feasible embedding:
\begin{definition}[$\VNEP$, Decision Variant]
Given is a single request $\VG$ that shall be embedded on the substrate graph $\SG$. The task is to find any feasible embedding or to decide that no feasible embedding exists.
\end{definition}

\subsection{Variants of the $\VNEP$ \& Nomenclature}
As discussed when reviewing the related work in Section~\ref{sec:related-work}, additional requirements are enforced in many settings. Accordingly, we now formalize (i) node placement, (ii) edge routing, and (iii) latency restrictions. Node placement and edge routing restrictions effectively exclude potential mapping options for nodes and edges. For latency restrictions we introduce latency bounds for each of the virtual edges. 

\begin{definition}[Node Placement Restrictions]
For each virtual node $i \in \VV$ a set of forbidden substrate nodes $\VVlocForbidden \subset \SV$ is provided. Accordingly, the set of allowed nodes $\VVloc$ is defined to be $\{u \in \SV \setminus \VVlocForbidden \,|\,\Scap(u) \geq \Vcap(i) \}$.
\end{definition}

\begin{definition}[Routing Restrictions]
For each virtual edge $(i,j) \in \VE$ a set of forbidden substrate edges $\VElocForbidden \subseteq \SE$ is provided. Accordingly, the set of allowed edges $\VEloc$ is set to be $\{(u,v) \in \SE \setminus \VElocForbidden \,|\,\Scap(u,v) \geq \Vcap(i,j) \}$.
\end{definition}

\begin{definition}[Latency Restrictions]
For each substrate edge $e \in \SE$ the edge's latency is given via $\Slat(e) \in \mathbb{R}_{\geq 0}$. Latency bounds for virtual edges are specified via the function \mbox{$\Vlat: \VE \to \mathbb{R}_{\geq 0} \cup \{\infty\}$}, such that the latency along the substrate path $\mapE(i,j)$, used to realize the edge \mbox{$(i,j) \in \VE$}, is less than $\Vlat(i,j)$. Formally, the definition of feasible embeddings (cf. Definition~\ref{def:feasible-embedding}) is extended by including that $\sum_{e \in \mapE(i,j)} \Slat(e) \leq \Vlat(i,j)$  holds for $(i,j) \in \VE$.
\end{definition}

We introduce the following taxonomy to denote the different problem variants.
\begin{definition}[Taxonomy]
\label{def:nomenclature}
We use the notation $\var C A $ to indicate whether and which of the capacity constraints $\varStyle{C}$ and  which of the additional constraints $\varStyle{A}$ are enforced.
\begin{description}
\item[$\varStyle{C}$] We denote by  $\varStyle{V}$ node capacities, by $\varStyle{E}$ edge capacities, and by $\varStyle{-}$ that none are used. When node or edge capacities are \emph{not} considered, we set the capacities of the respective substrate elements to $\infty$.
\item[$\varStyle{A}$] For the additional restrictions \varStyle{-}, \varStyle{N}, \varStyle{L}, and \varStyle{R} stand for no restrictions, node placement, latency, and routing restrictions, respectively.
\end{description}
\vspace{-12pt}
\end{definition}
Hence, $\var VE - $ indicates the classic $\VNEP$ without additional constraints while obeying capacities and $\var - NL $ indicates the combination of node placement and latency restrictions without considering substrate capacities. We note that the introduction of more restrictions only makes the respective problem harder:
\begin{lemma}
\label{lem:vnep-only-gets-harder}
A VNEP variant $\var A C $ that encompasses all restrictions of $\var A' C' $ is at least as hard as $\var A' C' $.
\end{lemma}
\begin{proof}
The capacity constraints as well as the additional requirements were all formulated in such a fashion that any one of these can be disabled. Considering capacities and latencies, one may set the respective substrate capacities to $\infty$ and the latencies of edges to $0$, respectively. For node placement and edge restrictions one may set the forbidden node and edge sets to the empty set. Hence, there exists a trivial reduction from $\var A C $ to $\var A' C' $ and the result follows.
\end{proof}

\subsection{Relaxing Constraints}

Within this work, we show the $\VNEP$ to be $\NPcomplete$ under many meaningful restriction combinations. This in turn also implies the inapproximability of the respective $\VNEP$ variants (unless $\compPeqNP$). Hence, it is natural to consider a broader class of (approximation) algorithms that may violate constraints by a certain factor: instead of answering the question whether a valid embedding exists that satisfies all capacity constraints, one might for example seek an embedding that uses at most two times the actual capacities. We refer to these embeddings as approximate embeddings:
\begin{definition/}[$\alpha$- / $\beta$- / $\gamma$-Approximate Embeddings]~\\
\label{def:approximate-embeddings}
A mapping $\map$ is an approximate embedding, if it is valid and violates capacity or latency constraints only within a certain bound. Specifically, we call an embedding $\alpha$- and $\beta$-approximate, when node and edge allocations are bounded by $\alpha$ and $\beta$ times the respective node or edge capacity. Considering latency restrictions, we call a mapping $\gamma$-approximate when latencies are within a factor of $\gamma$ of the original bound. Formally, the following must hold for $\alpha, \beta,\gamma \geq 1$:
\begin{alignat*}{7}
A_m(u) & \,\leq \,&& \alpha \cdot \Scap(u) & \forall u \in \SV \\
A_m(u,v) & \leq && \beta \cdot \Scap(u,v) \quad & \forall (u,v) \in \SE \\
\sum_{e \in \mapE(i,j)} \hspace{-10pt} \Slat(e) & \leq  && \gamma \cdot \Vlat(i,j) &  \forall (i,j) \in \VE \tag*{$\square$}
\end{alignat*}

\end{definition/}

\subsection{Integer Programming Formulation}
\label{sec:ip-formulation}

We now give an Integer Programming (IP) formulation, which can be used to solve any of the considered decision $\VNEP$ variants. A similar formulation was proposed in~\cite{infuhr2011introducing}.
Given the  hardness results presented in this paper and given that solving IPs lies in $\compNP$~\cite{papadimitriou1981complexity}, the IP may serve as an attractive approach to solve the respective variants in \emph{exponential} time.
 Besides the practical application, the existence of our formulation (constructively) shows that the $\VNEP$ variants considered here are also all contained in $\compNP$. 
 
 Our formulation naturally encompasses node placement and routing restrictions, while for latencies an additional constraint is introduced.
The decision variable $x \in \{0,1\}$ is used to indicate, whether the request graph $\VG$ is embedded or not. By maximizing $x$, the IP decides whether a feasible embedding exists ($x=1$) or whether no such embedding exists (\mbox{$x=0$}). The mapping of virtual nodes is modeled using decision variables $y^u_i \in \{0,1\}$ for $i \in \VV$ and $u \in \SV$. If $y^u_i = 1$ holds, then the virtual node $i \in \VV$ is mapped on substrate node $u \in \SV$. Constraint~\ref{alg:ip:node-embedding} enforces that each virtual node is mapped to one substrate node, \emph{if} the request is embedded $(x=1)$, while Constraint~\ref{alg:ip:node-forbidding} excludes unsuitable substrate nodes.

For computing edge mappings the decision variables \mbox{$z^{u,v}_{i,j} \in \{0,1\}$} for $(i,j) \in \VE$ and $(u,v) \in \SE$ are employed. If $z^{u,v}_{i,j} = 1$ holds, then the substrate edge $(u,v)$ lies on the path $\mapE(i,j)$.
Constraints~\ref{alg:ip:edge-embedding} and \ref{alg:ip:edge-forbidding} embed virtual links as paths in the substrate, if the request is embedded. In particular, Constraint~\ref{alg:ip:edge-embedding} constructs a unit flow for virtual edge $(i,j) \in \VE$ from the location $u \in \SV$ onto which $i$ was mapped ($y^u_i = 1$)  to the location $v \in \SV$ onto which $j$ was mapped ($y^{v}_j = 1$), while Constraint~\ref{alg:ip:edge-forbidding} excludes unsuitable edges. 
Constraints~\ref{alg:ip:node-capacity} and \ref{alg:ip:edge-capacity} enforce that substrate capacities are obeyed. 
Lastly, Constraint~\ref{alg:ip:latency} is only used when latencies are considered: it enforces that the sum of latencies along the embedding path of a virtual edge is smaller than the respective latency bound.

{

 \newcolumntype{F}{>{$\displaystyle\,}r<{$}@{\hspace{0.0em}}}
 \newcolumntype{C}{>{$\displaystyle\,}c<{$}@{\hspace{0.0em}}}
 \newcolumntype{B}{>{$\displaystyle\,}r<{$}@{\hspace{0.0em}}}
 \newcolumntype{R}{>{$\displaystyle}r<{$}@{\hspace{0.2em}}}
 \newcolumntype{S}{>{$\displaystyle}r<{$}@{\hspace{0.2em}}}
 \newcolumntype{L}{>{$\displaystyle}l<{$}@{\hspace{0.2em}}}
 \newcolumntype{Q}{>{$\displaystyle}l<{$}@{\hspace{0.3em}}}
 \newcommand{\muRow}[1]{\multirow{2}{*}{$\displaystyle #1 $}}
 \newcommand{\muRowThree}[1]{\multirow{3}{*}{$\displaystyle #1 $}}
 \newcommand{\muColL}[1]{\multicolumn{1}{L}{$\displaystyle #1 $}}
 \newcommand{\muColC}[1]{\multicolumn{1}{c}{$\displaystyle #1 $}}
 \newcommand{\muColR}[1]{\multicolumn{1}{R}{#1}}

 \newcounter{IPnumber}
 \setcounter{IPnumber}{0}
 \newcommand{\tagIt}[1]{\refstepcounter{equation}\textnormal{({\theequation})} \label{#1}}
  \newcommand{\tagItLat}[1]{\refstepcounter{equation}\textnormal{$\hspace{0pt}^\star\textnormal{(\theequation)}$} \label{#1}}

{
 \LinesNotNumbered
 \renewcommand{\arraystretch}{1.25}

 \begin{IPFormulation}{t!}
 \popline
 \popline

 \SetAlgorithmName{Integer Program}{}{{}}

 \scalebox{0.96}{

 \newcommand{\spaceIt}{\qquad\quad\quad}
 \newcommand{\miniSpace}{\hspace{1.5pt}}
 \begin{tabular}{FRLQB}
\multicolumn{4}{r}{\parbox{0.95\columnwidth}{~}} \\[-18pt]

\multicolumn{4}{C}{ \phantomsection \textnormal{max~}  x ~~~~~~~~~~~~~~}  & \tagIt{alg:ip:obj} \\

\phantomsection \sum_{u \in \SV } y^u_{i} & = & x & \forall i \in \VV &  \tagIt{alg:ip:node-embedding} \\

\phantomsection  \sum_{u \in \SV \setminus \VVloc} \hspace{-6pt} y^u_{i} & = & 0 & \forall i \in \VV &  \tagIt{alg:ip:node-forbidding} \\

\phantomsection  \sum_{(u,v) \in \delta^+(u)} \hspace{-12pt} z^{u,v}_{i,j} - \hspace{-9pt} \sum_{(v,u) \in \delta^-(u)} \hspace{-12pt} z^{v,u}_{i,j} & = & y^u_{i} - y^u_{j} ~~& \forall (i,j) \in \VE, u \in \SE &  \tagIt{alg:ip:edge-embedding} \\

\phantomsection  \sum_{(u,v) \in \SE \setminus \VEloc} \hspace{-8pt} z^{u,v}_{i,j} & = & 0 & \forall (i,j) \in \VE &  \tagIt{alg:ip:edge-forbidding} \\

\phantomsection  \sum_{i \in \VV} \Vcap(i) \cdot y^u_{i} & \leq & \Scap(u) & \forall u \in \SV & \tagIt{alg:ip:node-capacity} \\

\phantomsection  \sum_{(i,j) \in \VE} \Vcap(i,j) \cdot z^{u,v}_{i,j} & \leq & \Scap(u,v) & \forall (u,v) \in \SE & \tagIt{alg:ip:edge-capacity} \\

\phantomsection  \sum_{(u,v) \in \SE} \Slat(u,v) \cdot z^{u,v}_{i,j} & \leq & \Vlat(i,j) & \forall (i,j) \in \VE & \tagItLat{alg:ip:latency} \\
	 	
 \end{tabular}
 }
 \caption{$\VNEP$ Decision Variant}
 \label{alg:ip:formulation}

 \end{IPFormulation}

 }

}

\section{Reduction Framework}
\label{sec:reduction-framework}

{

\newcommand{\formula}{\ensuremath{\phi}}

\newcommand{\Lit}{\ensuremath{\mathcal{L_{\formula}}}}
\newcommand{\LitI}[1][i]{\ensuremath{\mathcal{L}_{#1}}}
\newcommand{\LitIpre}[1][i]{\ensuremath{\mathcal{L}_{#1,\mathrm{pre}}}}
\newcommand{\LitIJ}[1][i,j]{\ensuremath{\mathcal{L}_{#1}}}

\newcommand{\Cla}{\ensuremath{\mathcal{C}_{\formula}}}

\newcommand{\lit}[1][k]{\ensuremath{x_{#1}}}
\newcommand{\cla}[1][i]{\ensuremath{\mathcal{C}_{#1}}}

\newcommand{\FG}{\ensuremath{G_{\formula}}}
\newcommand{\FV}{\ensuremath{V_{\formula}}}
\newcommand{\FE}{\ensuremath{E_{\formula}}}

\renewcommand{\req}{\ensuremath{r({\formula})}}

\renewcommand{\SG}{\ensuremath{G_{S(\formula)}}}
\renewcommand{\SV}{\ensuremath{V_{S(\formula)}}}
\renewcommand{\SE}{\ensuremath{E_{S(\formula)}}}

\newcommand{\AssI}[1][i]{\ensuremath{\mathcal{A}_{#1}}}
\newcommand{\ass}[1][i,m]{\ensuremath{a_{#1}}}

\newcommand{\assG}{\ensuremath{\alpha}}

\newcommand{\True}{\ensuremath{\mathrm{T}}}
\newcommand{\False}{\ensuremath{\mathrm{F}}}

This section presents the main insight and contribution of our paper, namely a generic reduction framework that allows to derive hardness results by slightly tailoring the proof for the individual problem variants. 
Our reduction framework relies on $\ThreeSAT$ and we first introduce some notation. Afterwards we continue by constructing a (partial) $\VNEP$ instance, whose solution will indicate whether the $\ThreeSAT$ formula is satisfiable. 

\subsection{$\ThreeSAT$: Notation and Problem Statement}
We denote by $\Lit = \{\lit\}_{k \in [N]}$ a set of $N\in \mathbb{N}$ literals and by $\Cla = \{\cla\}_{i \in [M]}$ a set of $M\in \mathbb{N}$ clauses, in which literals may occur either positively or negated.
The formula $\phi = \bigwedge_{\cla \in \Cla} \cla$ is a $\ThreeSAT$ formula, iff. each clause $\cla$ is the disjunction of at most 3 literals of $\Lit$. Denoting the truth values by $\False$ and $\True$, $\ThreeSAT$ asks to determine whether an assignment $\alpha : \Lit \to \{\False, \True\}$ exists, such that $\phi$ is satisfied. $\ThreeSAT$ is one of Karp's 21 $\NPcomplete$ problems:
\begin{theorem}[Karp~\cite{karp1972reducibility}]
Deciding $\ThreeSAT$ is $\NPcomplete$.
\end{theorem}

For reducing $\ThreeSAT$ to $\VNEP$, it is important that the clauses be ordered and we define the following:
\begin{definition/}[First Occurence of Literals]
We denote by $\mathcal{C}: \Lit \to [M]$ the function yielding the index of the clause in which a literal first occurs. Hence, if $\mathcal{C}(\lit) = i$, then $\lit$ is contained in $\cla[i]$ while not contained in $\cla[i']$ for $i' \in [i-1]$.\hspace*{\fill} $\square$
\end{definition/}

As we are interested the satisfiability of a $\ThreeSAT$ formula $\formula$, we define the set of satisfying assignments \emph{per clause}:
\begin{definition}[Satisfying Assignments]
We denote by \mbox{$\AssI = \{\ass : \LitI \to \{\False,\True\}  ~|~\ass \textnormal{ satisfies } \cla \}$} the set of all possible assignments of truth values to the literals $\LitI$ of $\cla$ satisfying $\cla$. Note that all elements of $\AssI$ are functions.
\end{definition}
\begin{figure}[tb!]
\centering
\includegraphics[width=1\columnwidth]{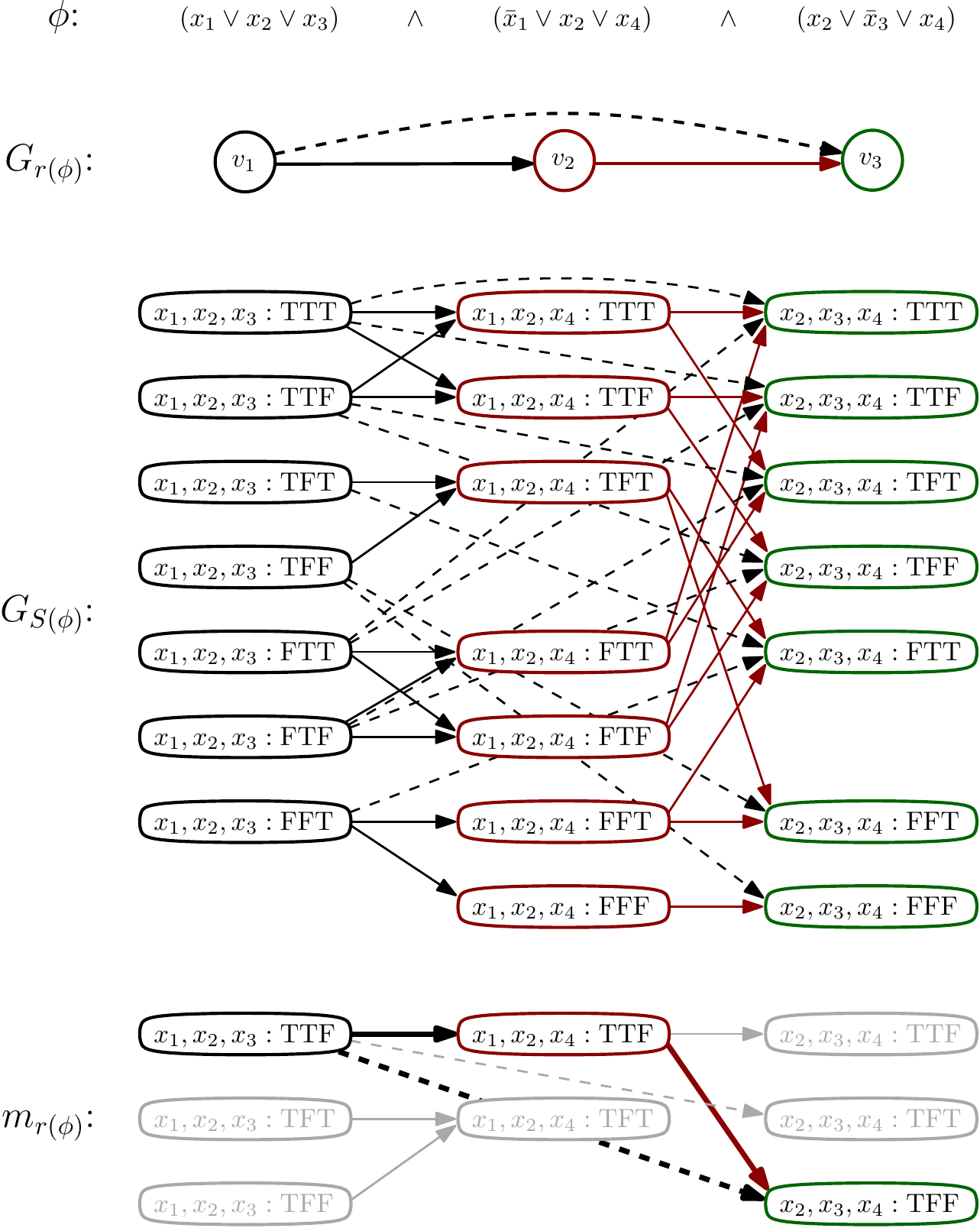}
\caption{Visualization of the construction of substrate and request graphs for the $\ThreeSAT$ formula $\formula$ (cf.~Definitions~\ref{def:substrate-construction} and~\ref{def:request-construction}). Additionally, a mapping $\map$ satisfying the conditions of Lemma~\ref{lem:base-lemma} is shown. Accordingly, the formula $\phi$ is satisfied. Concretely, the mapping represents the assignment of truth values $x_1= \mathrm{T}$, $x_2 = \mathrm{T}$, $x_3 = \mathrm{F}$, $x_4 = \mathrm{F}$.}
\label{fig:3-sat-illustration-new}
\end{figure}

Lastly, to abbreviate notation, we employ $\LitIJ = \LitI \cap \LitI[j]$ to denote the intersection of the literal sets of $\cla$ and $\cla[j]$.

\subsection{General $\VNEP$ Instance Construction}
\label{sec:vnep-instance-construction}

For a given $\ThreeSAT$ formula $\formula$, we now construct a $\VNEP$ instance consisting of a substrate graph $\SG$ and a request graph $\VG$. The question whether the formula $\formula$ is satisfiable will eventually reduce to the question whether a feasible embedding of $\VG$ on $\SG$ exists. Figure~\ref{fig:3-sat-illustration-new} illustrates the construction described in the following.

\begin{definition}[Substrate Graph $\SG$]
\label{def:substrate-construction}
For a given $\ThreeSAT$ formula $\formula$ we define the substrate graph $\SG=(\SG,\SE)$ as follows.
For each clause $\cla \in \Cla$ and each potential assignment of truth values satisfying $\cla$, a substrate node is constructed, i.e. we set $\SV = \bigcup_{\cla \in \Cla} \AssI$. We connect two  substrate nodes $a_{i,m} \in \SV$ and $a_{j,n} \in \SV$, iff. a literal $\lit$ is introduced in the clause $\cla$ for the first time and is also used in clause $\cla[j]$, \emph{and} $a_{i,m}$ and $a_{j,n}$ agree on the truth values of the literals contained in both clauses. Formally, we set:
\[
\SE = \left\{ (a_{i,m}, a_{j,n}) \,\middle| \begin{array}{l}
 \exists \lit \in \LitIJ \textnormal{ with } \mathcal{C}(\lit) = i  \textnormal{ and } \\
 a_{i,m}(x_l) = a_{j,n}(x_l) \textnormal{ for } x_l \in \LitIJ
\end{array} \right\}
\]
Capacities etc. are introduced in the respective reductions.
\end{definition}

\begin{definition}[Request Graph $\req$]
\label{def:request-construction}
For a given $\ThreeSAT$ formula $\formula$ we define the request graph $\VG = (\VV,\VE)$ as follows.
For each clause $\cla \in \Cla$ a node $v_i$ is introduced, i.e. $\VV = \{v_i \,|\, \cla \in \Cla\}$. Matching the construction of the substrate graph $\SG$, we introduce directed edges $(v_i,v_j) \in \VE$ only if there exists a literal $\lit \in \cla$ being introduced in $\cla$ and being also used in the clause $\cla[j]$:
\[
\VE = \{(v_i, v_j) \,|\, \exists \lit \in \LitIJ \textnormal{ with } \mathcal{C}(\lit) = i\}
\]
Demands etc. are introduced in the respective reductions.
\end{definition}

\subsection{The Base Lemma}
\label{sec:base-lemma}

Nearly all of our results are based on the following lemma.

\begin{lemma}
\label{lem:base-lemma}
The $\ThreeSAT$ formula $\formula$ is satisfiable
\begin{center}
 if and only if
\end{center}
 
\noindent there exists a valid mapping $\map$ of $\VG$ on $\SG$, such that 
\begin{enumerate}
\item each virtual node $v_i$ is mapped on a substrate node corresponding to assignments $\AssI$ of the $i$-th clause, i.e. $\mapV(v_i) \in \AssI$ holds for all $v_i \in \VV$, and
\item virtual edges are embedded using a single substrate edge, i.e. $|\mapE(v_i,v_j)| = 1$ holds for all $(v_i,v_j) \in \VE$.
\end{enumerate}

\end{lemma}
\begin{proof}

We first show that if $\formula$ is satisfiable, then such a mapping $\map$ must exist. Afterwards, we show that if such a mapping $\map$ exists, then $\formula$ must be satisfiable.

Assume that $\formula$ is satisfiable and let $\assG: \Lit \to \{  \mathrm{F}, \mathrm{T} \}$ denote an assignment of truth values, such that $\assG$ satisfies $\formula$. We construct a mapping $\map = (\mapV, \mapE)$ for request $\req$ as follows.
The virtual node $v_i \in \VV$ corresponding to clause $\cla$ is mapped onto the substrate node $\ass[i,m] \in \AssI \subseteq \SV$, iff. $\ass[i,m]$ agrees with $\assG$ on the assignment of truth values to the contained literals, i.e. $\ass[i,m](x_k) = \alpha(x_k)$ for $x_k \in \cla$. As $\assG$ satisfies $\formula$, it satisfies each clause and hence $\mapV(v_i) \in \SV$ holds for all $\cla \in \Cla$.
The virtual edge $(v_i,v_j) \in \VE$ is mapped via the direct edge between $\mapV(v_i)$ and $\mapV(v_j)$. This edge $ (\mapV(v_i), \mapV(v_j))$ must exist in $\SE$, as the existence of virtual edge $(v_i,v_j)$ implies that clause $\cla$ is the first clause introducing a literal of $\LitIJ$ and  $\mapV(v_i) = \ass[i,m]$ and $\mapV(v_j) = \ass[j,n]$ must agree by construction on the assignment of truth values for all literals.
Clearly, the constructed mapping $\map$ fulfills both the conditions stated in the lemma, hence completing the first half of the proof.

We now show that if there exists a mapping $\map$ meeting the two requirements stated in the lemma, then the formula $\formula$ is indeed satisfiable. We constructively recover an assignment of truth values $\assG : \Lit \to \{ \mathrm{F}, \mathrm{T} \}$ from the mapping $\map$ by iteratively extending the initially empty assignment.
Concretely, we iterate over the mappings of the virtual nodes corresponding to clauses $\Cla$ one by one (according to the precedence relation of the indices). By our assumption on the node mapping, $\mapV(v_i) \in \AssI$ holds. Accordingly, as the substrate node $\mapV(v_i)$ represents an assignment of truth values to the literals of clause $\cla$, we extend $\assG$ by setting $\assG(\lit) \triangleq \big[\mapV(v_i)\big](\lit)$ for all literals $\lit$ contained in $\cla$.

We first show that this extension is always valid in the sense that previously assigned truth values are never changed. To this end, assume that the clauses $\cla[1], \cla[2], \dots, \cla[i-1]$ were handled without any such violations. Hence the literals $\bigcup_{j < i} \LitI[j]$ have been assigned truth values in the first $i-1$ iterations not contradicting previous assignments. 
When extending $\assG$ by the mapping of $\mapV(v_i)$ in the $i$-th iteration, there are two cases to consider. First, if none of the literals $\LitI$ were previously assigned a truth value, i.e. $\LitI \cap \bigcup_{j < i} \LitI[j] = \emptyset$ holds, then the extension of $\assG$ as described above cannot lead to a contradiction. Otherwise, if $\LitIpre = \LitI \cap \bigcup_{j < i} \LitI[j] \neq \emptyset$ holds, we show that extending $\assG$ by $\mapV(v_i) = \ass$ does not change the truth value of any literal $x_k$ contained in $\LitIpre$. 

For the sake of contradiction, assume that $\lit \in \cla$ is a literal, for which $\assG(x_k)$ does not equal $\big[\mapV(v_i)\big](x_k)$. As $\lit$ was previously assigned a value, there must exist a clause $\cla[j]$ in which $\lit$ was first used, such that $j<i$ holds. Let $\mapV(v_i)=\ass[i,m] \in \AssI[i]$ and $\mapV(v_j)=\ass[j,n] \in \AssI[j]$. As the edge $(v_j,v_i)$ is contained in $\VE$ by definition and all edges are mapped using a single substrate edge by our assumptions, $\mapE(v_i,v_j) =  \langle (\ass[j,n], \ass[i,m]) \rangle$ must hold. Hence, as $(\ass[j,n], \ass[i,m]) \in \SE$ must hold and edges are only introduced \emph{if} assignments agree with each other, we have $\big[\mapV(v_j)\big](x_k) = \ass[j,n](x_k) = \ass[i,m](x_k) = \big[\mapV(v_i)\big](x_k)$. This contradicts our assumption that $\assG(x_k) \neq \big[\mapV(v_i)\big](x_k)$ holds. Hence, the extension of $\assG$ is always valid.

By construction of the substrate graph $\SG$, the node set $\AssI \subseteq \SV$ contains only the assignments of truth values for the literals $\LitI$ of clause $\cla \in \Cla$ that satisfy the respective clause. Hence, $\assG$ satisfies all of the clauses and hence satisfies $\formula$, completing the proof of the base lemma.
\end{proof}

The base lemma is the heart of our reduction framework for obtaining our results and we note that the construction of the substrate and the request graph is polynomial in the size of the $\ThreeSAT$ formula. Indeed, the base lemma forms the basis for polynomial-time reductions for the different $\VNEP$ decision variants. Concretely, consider some $\VNEP$ variant $\var X Y $. If this variant is `expressive' enough such that any feasible embedding must meet the criteria of Lemma~\ref{lem:base-lemma}, then $\var X Y $ is -- by reduction from $\ThreeSAT$ -- $\NPhard$. Furthermore, the Integer Program presented in Section~\ref{sec:ip-formulation} shows that all the $\VNEP$ variants considered here lie in $\compNP$ and hence the successful application of the base lemma shows the $\NPcompleteness$ of the respective variant.
As a result, for the considered $\VNEP$ variant, any optimization problem (e.g. cost) cannot be approximated within any factor. The following lemma formalizes this observation:
\begin{lemma}
\label{lem:theorem-factory}
If there is a polynomial-time reduction from $\ThreeSAT$ to the $\VNEP$ decision problem under constraints $\var X Y $, then the respective $\VNEP$ variant is $\NPcomplete$. Furthermore, any optimization problem over the same set of constraints is (i) $\NPhard$ and (ii) inapproximable (within any factor), unless $\compPeqNP$ holds.
\end{lemma}

Lastly, the following lemma will prove useful when applying the base lemma.

\begin{lemma}
\label{lem:request-normalization}
Exactly one of the following two 
following properties holds for formula~$\varphi$:
\begin{enumerate}
\item The clauses of $\formula$ can be ordered such that within the corresponding request graph $\VG$ only the node $v_1 \in \VV$ has no incoming edges.
\item $\varphi$ can be decomposed into formulas $\varphi_1$ and $\varphi_2$, such that the sets of literals occurring in $\varphi_1$ and $\varphi_2$ are disjoint, while $\varphi = \varphi_1 \wedge \varphi_2$ holds. Hence, $\varphi$ is satisfiable iff. $\varphi_1$ and $\varphi_2$ are (independently) satisfiable.
\end{enumerate}
\end{lemma}

\begin{proof}
We prove the statement by a greedy construction and assume that the clauses are initially unordered. 
We iteratively assign an index to the clauses, keeping track of which clauses were not assigned an index yet.
Initially, pick any of the clauses and assign it the index $1$. Now, iteratively choose any clause which contains a literal that already occurs in the set of indexed clauses. If no such clause exists, then the clauses already indexed and the clauses not indexed obviously represent a partition of the literal set and hence the second statement holds true. 
However, if the greedy step succeeded every time, then the following holds with respect to the constructed ordering: any virtual node $v_i$ corresponding to clause $\cla$, for $i > 1$, must have an incoming edge by Definition~\ref{def:request-construction} as the clause overlapped with the already introduced literals. 
\end{proof}

\section{$\NPcompleteness$ of the $\VNEP$}

\label{sec:np-completeness-of-the-vnep}

We employ our framework outlined in the previous section to derive
 a series of hardness results for the $\VNEP$. In particular, we first show the $\NPcompleteness$ of the original $\VNEP$ variant $\var VE - $ in the absence of additional restrictions. Given this result, we investigate several other problem settings and show, among others, that also deciding $\var - LN $ is $\NPcomplete$. Hence, even when the physical network 
 does not impose any resource constraints (i.e., nodes and links have infinite
 capacities), finding an embedding satisfying latency and node placement restrictions is $\NPcomplete$. Again, it must be noted that adding further restrictions only renders the $\VNEP$ harder (cf. Lemma~\ref{lem:vnep-only-gets-harder}).

\subsection{$\NPcompleteness$ under Capacity Constraints}

We first consider the most basic $\VNEP$ variant $\var VE - $.

\begin{theorem}
\label{thm:vnep-capacity-nodes-and-edges-is-np-complete}
$\VNEP$ $\var VE - $ is $\NPcomplete$ and cannot be approximated under any objective (unless $\compPeqNP$). 
\end{theorem}
\begin{proof}
We show the statement via a polynomial-time reduction from $\ThreeSAT$ according to Lemmas~\ref{lem:base-lemma} and \ref{lem:theorem-factory}. Given is a $\ThreeSAT$ formula $\formula$. We assume for now that the first statement of Lemma~\ref{lem:request-normalization} holds, i.e. that within the request graph $\VG$ only the first node $v_1 \in \VV$ has no incoming edge.

To enforce the properties of Lemma~\ref{lem:base-lemma}, we set the substrate and request capacities for some small $\lambda$, $0 < \lambda < 1/|\Cla|$, as follows.
The capacity of substrate nodes is determined by the clause whose assignments they represent. Furthermore, the capacities decrease monotonically with each clause. Similarly, but now increasing per clause, the capacities of edges are determined by the clause that the edge's head corresponds to:
\begin{alignat*}{7}
\Scap(\ass[i,m]) &\,= \,&& 1 + \lambda \cdot (M - i) ~~&&\forall \,\cla \in \Cla, \ass[i,m] \in \AssI \\
 \Scap(e) &\,=\, && 1 + \lambda \cdot i && \forall \, \cla \in \Cla, e \in \delta^-(\AssI)
\end{alignat*}

The demands are set to match the respective capacities:
\begin{alignat*}{7}
~~~ \Vcap(v_i) &\,=\,&& 1 + \lambda \cdot (M - i) ~~&&\forall \,v_i \in \VV \\
 \Vcap(e) &\,=\,&& 1 + \lambda \cdot i && \forall \, v_j \in \VV, e \in \delta^-(v_j)
\end{alignat*}

Due to the decreasing node demands and capacities, virtual node $v_j \in \VV$ corresponding to clause $\cla$ can only be mapped on substrate nodes $\bigcup^j_{k=1} \AssI[k]$. Due to the choice of $\lambda$, the capacity of any substrate node is less than $2$ while each virtual node has a demand larger than $1$. Hence,  two virtual nodes can never be collocated (mapped) on the same substrate node. Thus, all virtual edges must be mapped onto at least a single substrate edge. Considering the virtual edge $e=(v_i,v_j) \in \VE$ with demand $\Vcap(e) = 1 + \lambda \cdot j$, the virtual node $v_j$ must be mapped on a substrate node having an incoming edge of at least capacity $ 1+ \lambda \cdot j$. As the edge capacities increase with the clause index, only the substrate nodes in $\bigcup^M_{k=j} \AssI[k]$ satisfy this condition. Hence, if node $v_j$ has an incoming edge, it can only be mapped on nodes in $\bigcup^j_{k=1} \AssI[k] \cap \bigcup^M_{k=j} \AssI[k] = \AssI[j]$. As we assumed that the first statement of Lemma~\ref{lem:request-normalization} holds for $\formula$ and hence all nodes $v_2,\ldots,v_M$ have an incoming edge, we obtain that the virtual node $v_i$ must be mapped on $\AssI \subseteq \SV$ for $i = 2,\ldots,M$. Considering the first node $v_1$, we observe that only nodes in $\AssI[1]$ offer sufficient capacity to host $v_1$. Hence, any feasible embedding will obey the first statement of Lemma~\ref{lem:base-lemma} regarding the node mappings.

We now show that \emph{any} feasible mapping will also obey the second property of Lemma~\ref{lem:base-lemma}, namely, that any virtual edge is mapped on exactly one substrate edge. To this end, assume for the sake of contradiction that $(v_i,v_j) \in \VE$ is not mapped on a single substrate edge. As $v_i$ must be mapped on some node $a_{i,m} \in \AssI$ and $v_j$ must be mapped on some node $a_{j,n} \in \AssI[j]$, and as both the request and the substrate are directed acyclic graphs, the mapping of edge $(v_i,v_j)$ must route through at least one intermediate node. Denote by $a_{k,l} \in \AssI[k]$ for $i < k < j$ the first intermediate node via which the edge $(v_i,v_j)$ is routed. By construction, the capacity of the substrate edge $(a_{i,m}, a_{k,l})$ is $1 + \lambda \cdot k$. However, as $k < j$ holds and the edge $(v_i,v_j)$ has a demand of $1 + \lambda \cdot j$, the edge $(v_i,v_j)$ cannot be routed via $a_{k,l}$. Hence, the only feasible edges for embedding the respective virtual edges are the direct connections between any two substrate nodes. 

Therefore, all \emph{feasible} solutions indicate the satisfiability of the formula $\formula$. Any algorithm computing a feasible solution to the $\VNEP$ obeying node and edge capacities, decides $\ThreeSAT$.

Lastly, we argue for the validity of our assumption on the structure of $\formula$, namely that the first statement of Lemma~\ref{lem:request-normalization} holds. If this were not to hold, then the second statement of Lemma~\ref{lem:request-normalization} holds true and the formula can be decomposed (potentially multiple times) into disjoint subformulas $\varphi_1,\ldots, \varphi_k$, such that (i) $\varphi = \bigwedge^k_{i=1} \varphi_i$ holds, and (ii) such that the first condition of Lemma~\ref{lem:request-normalization} holds for each subformula. Accordingly, assuming that an algorithm exists which can construct feasible embeddings whenever they exist, this algorithm can be used to decide the satisfiability of each subformula, hence deciding the original satisfiability problem.
\end{proof}

\subsection{$\NPcompleteness$ under Additional Constraints}

Building on the above $\NPcompleteness$ proof, we can adapt it easily to other settings.

\begin{theorem}
\label{thm:decision-edge-capacities-node-placement}
$\VNEP$ $\var E N $ is $\NPcomplete$ and cannot be approximated under any objective (unless $\compPeqNP$). 
\end{theorem}
\begin{proof}
In this setting node placement restrictions and substrate edge capacities are enforced.
We apply the same construction as in the proof of Theorem~\ref{thm:vnep-capacity-nodes-and-edges-is-np-complete}. Employing the node placement restrictions, we can force the mapping of virtual node $v_i \in \VV$ onto substrate nodes $\AssI$ by setting $\VVlocForbidden[v_i] = \SV \setminus \AssI$ for all $v_i \in \VV$. By the same argument as before, virtual edges cannot be mapped onto paths as the intermediate nodes do not support the respective demand.
\end{proof}

\begin{theorem}
\label{thm:node-capacities-routing-restrictions}
$\VNEP$ $\var V R $ is $\NPcomplete$ and cannot be approximated under any objective (unless $\compPeqNP$).
\end{theorem}
\begin{proof}
In this setting only node capacities must be obeyed, while routing restrictions may be introduced.
We employ the same node capacities as in the proof of Theorem~\ref{thm:vnep-capacity-nodes-and-edges-is-np-complete}, such that virtual node $v_i \in \VV$ may only be mapped on nodes $\bigcup^i_{k=1} \AssI[k]$. Routing restrictions are set to only allow direct edges, i.e. $\VElocForbidden[v_i,v_j] = \SE \setminus (\AssI[i] \times \AssI[j])$ holds for each $(v_i,v_j) \in \VE$. Again, $v_1 \in \VV$ must be mapped on a node in $\AssI[1]$, while all other virtual nodes have at least one incoming edge according to Lemma~\ref{lem:request-normalization}. As multiple virtual nodes cannot be placed on the same substrate node and virtual edges must span at least one substrate edge, any node $v_j$ can only be mapped on nodes in $\AssI[j]$ for $j \in \{2,\ldots, M\}$. Together with the routing restrictions both requirements of Lemma~\ref{lem:base-lemma} are safeguarded and the result follows.
\end{proof}

\begin{theorem}
$\VNEP$s $\var - NR $ and $\var - NL $ are $\NPcomplete$ and cannot be approximated under any objective (unless $\compPeqNP$).
\label{thm:without-capacities-node-placement-and-routing-or-latencies}
\end{theorem}
\begin{proof}
In both cases capacities are not considered at all. Allowing for node placement restrictions, the first property of Lemma~\ref{lem:base-lemma} is easily safeguarded (cf. proof of Theorem~\ref{thm:decision-edge-capacities-node-placement}). By employing the same routing restrictions as in the proof of Theorem~\ref{thm:node-capacities-routing-restrictions} the result follows directly for the case $\var - NR $. Latency restrictions can be easily used to enforce that virtual edges do not span more than a single substrate edge. Concretely, we set unit substrate edge latencies and unit virtual edge latency bounds: if an edge was to be embedded via more than one edge, the latency restrictions would be violated. Hence, the result also holds for $\var - NL $.
\end{proof}

\section{$\NPcompleteness$ of Computing \\Approximate Embeddings} \label{sec:np-completeness-approximate-embeddings}

Given the hardness results presented in Section~\ref{sec:np-completeness-of-the-vnep}, the question arises to which extent the hardness can be overcome when only computing approximate embeddings (cf. Definition~\ref{def:approximate-embeddings}), i.e. embeddings that may violate capacity constraints or exceed latency constraints by certain factors. Based on the proofs presented in Section~\ref{sec:np-completeness-of-the-vnep}, we first derive hardness results for computing $\alpha$-approximate embeddings (allowing node capacity violations) and $\gamma$-approximate embeddings (allowing latency violations). For $\beta$-approximate embeddings, we give a reduction from a variant of the edge-disjoint paths problem.

\begin{theorem}
\label{thm:alpha-approximate-embeddings}
For $\var VE - $ and $\var V R $ finding an $\alpha$-approximate embedding is $\NPcomplete$ as well as inapproximable under any objective (unless $\compPeqNP$) for any $\alpha < 2 $.
\end{theorem}
\begin{proof}
Assume that there exists an algorithm computing $\alpha$-approximate embeddings for $\alpha = 2 - \varepsilon$, $0 < \varepsilon < 1$. We adapt the proofs of Theorem~\ref{thm:vnep-capacity-nodes-and-edges-is-np-complete} and \ref{thm:node-capacities-routing-restrictions} slightly. First, note that for $\alpha$-approximate mappings validity still has to hold according to Definition~\ref{def:approximate-embeddings}. Hence, by the decreasing node capacities the virtual node $v_j$ can only be mapped on substrate nodes $\bigcup^j_{k=1} \AssI[k]$. Furthermore, by either enforcing edge capacities or edge routing restrictions, the node $v_j$ can still only be mapped on $\AssI[j]$. Hence, the only missing piece to show that the respective proofs still hold is the fact that still at most a single virtual node can be mapped on a single substrate node. To ensure, that this still holds, we adapt the capacities. Concretely, we choose $\lambda$, such that $\lambda < \varepsilon/(2\cdot |\Cla|)$ holds. Hence, the capacity of any substrate node is less than $1+\varepsilon/2$. By relaxing the capacity constraints by the factor $2-\varepsilon$, the allowed substrate node allocations are upper bounded by $(1+\varepsilon/2)\cdot (2-\varepsilon) = 2 - \varepsilon - \varepsilon^2/2 < 2$. 
As the demand of any virtual node is larger than $1$, still at most a single virtual node can be mapped on a substrate node.
Hence, the respective proofs still apply and the results follow.
\end{proof}

Proving the $\NPcompleteness$ of $\gamma$-approximate embeddings goes along the same lines:

\begin{theorem}
\label{thm:gamma-approximate-embeddings}
For $\var - NL $ finding an $\gamma$-approximate embedding is $\NPcomplete$ as well as inapproximable under any objective (unless $\compPeqNP$) for any $\gamma < 2 $.
\end{theorem}
\begin{proof}
The proof of Theorem~\ref{thm:without-capacities-node-placement-and-routing-or-latencies} relied on the fact that due to the latency constraints each virtual edge must be mapped on a single substrate edge. As the latencies of substrate edges are uniformly set to 1 and all latency bounds are 1 as well, computing a $\gamma$-approximate embedding for $\gamma < 2$ implies that each virtual edge can still only be mapped on  a single substrate edge. Hence, the result of Theorem~\ref{thm:without-capacities-node-placement-and-routing-or-latencies} remains valid.
\end{proof}

For $\beta$-approximate embeddings we employ an inapproximability result of a variant of the edge-disjoint paths problem:

{
\renewcommand{\req}{\ensuremath{r(\text{dir})}}
\renewcommand{\SG}{\ensuremath{G}_{S(\text{dir})}}
\renewcommand{\SV}{\ensuremath{V}_{S(\text{dir})}}
\renewcommand{\SE}{\ensuremath{E}_{S(\text{dir})}}

\begin{definition}[$\DirEDPwC$~\cite{chuzhoy2007hardness}]
\label{def:DirEDPwC}
The Directed Edge-Disjoint Paths Problem with Congestion ($\DirEDPwC$) is defined as follows.
Given is a directed graph $G=(V,E)$ together with a set of $l \in \mathbb{N}$ source-sink pairs (commodities) $\{(s_k,t_k)\}_{k \in [l]}$, $s_k,t_k \in V$, and a constant $c \in \mathbb{N}$. The task is to find a path $P_k$ connecting $s_k$ to $t_k$ for each $k \in [l]$, such that at most $c$ many paths are routed via any edge $e \in E$.
\end{definition}

We show reductions from $\DirEDPwC$ to the $\VNEP$ variants $\var E N $ and $\var VE - $, respectively:

\begin{lemma}
\label{lem:DirEDPwC-to-var-E-N}
$\DirEDPwC$ can be reduced to the $\VNEP$ variant $\var E N $. The reduction preserves approximations.
\end{lemma}
\begin{proof}
Given a $\DirEDPwC$ instance, we employ its original graph $G=(V,E)$ as substrate, i.e. $\SG = G$, and define the request graph $\VG=(\VV,\VE)$ as follows. $\VV$ consists of two virtual nodes per commodity, \mbox{$\VV = \{i_k,j_k | k \in [l]\}$}, and we set \mbox{$\VE = \{(i_k,j_k) | k \in [l]\}$}. Let $\sigma: \VV \to \SV$ denote the function indicating the substrate node onto which the respective virtual node shall be mapped: $\sigma(i_k) = s_k$ and $\sigma(j_k) = t_k$ for all $k \in [l]$.
We set node mapping requirements such that the virtual nodes $i_k$ and $j_k$ must be mapped on $s_k$ and $t_k$, respectively: $\VVlocForbidden[i] = V \setminus \{\sigma(i)\}$ holds for $i \in \VV$. Setting edge capacities in the substrate to $c$ (the congestion value) and virtual edge demands to $1$, the equivalence of both problems under this reduction becomes apparent.
\end{proof}

\newcommand{\occurence}{\ensuremath{N}}

\begin{lemma}
\label{lem:DirEDPwC-to-var-VE--}
$\DirEDPwC$ can be reduced to the $\VNEP$ variant $\var VE - $. The reduction preserves approximations.
\end{lemma}
\begin{proof}
We use a similar construction as in the proof of Lemma~\ref{lem:DirEDPwC-to-var-E-N}. Now, instead of using node placement restrictions, we need to fix the virtual node mappings using a slightly different approach. Let $O: V \to \mathbb{N}$ denote the function counting how often a node occurs either as source or sink in the commodities. We adapt the substrate construction as follows. by the following rule. For each node $v \in \SV$, we add $O(v)$ many copies $\{v^1,v^2,\ldots,v^{O(v)}\}$ to the substrate graph and connect these to the original node $v$ of the substrate (in both directions). Now, let $U: V \to [|V|]$ denote any function that assigns each vertex a unique numeric identifier. We define substrate node capacities according to the following rule: all original nodes, $v \in \SV \cap V$, are assigned a capacity of $0$, while setting $\Scap(v^h) = U(v)$ for $v \in V$ and $h \in [O(v)]$. Accordingly, the demand of virtual nodes are set as follows: $\Vcap(i) = U(\sigma(i))$ for $i \in \VV$. We note the following properties: (i) the sum of available node capacities equals the sum of demanded capacities, and (ii) a virtual node $i \in \VV$ must be mapped on a substrate node $v^h \in \SV$ with $U(v^h) = U(\sigma(i))$, as otherwise node capacities cannot tally. As a copy $v^h \in \SV$ is only connected to its original node $v \in \SV \cap V$, any path starting at or leading to $v^h$ must be routed via $v$. Setting the edge capacities of original edges to $c$ while setting the capacity of edges incident to a copy $v^h$ to $1$, an equivalent $\var VE - $ instance is obtained. 
\end{proof}
}

It is well-known that the edge-disjoint paths problem on directed graphs is hard to approximate:
\begin{theorem}[Chuzhoy et al.~\cite{chuzhoy2007hardness}]
Let $n = |V|$ denote the number of nodes. The $\DirEDPwC$ is hard to approximate within $\Theta(\log n / \log \log n)$, unless $\compNP \subseteq \mathcal{BP}\text{-}\textit{TIME}(\bigcup_{d \geq 1} n^{d \log \log n})$ holds.
\end{theorem}

Above, $\mathcal{BP}\text{-}\textit{TIME}(f(n))$ denotes the class of problems solvable by probabilistic Turing machines in time $f(n)$ with bounded error-probability~\cite{arora2009computational}. Given the approximation-preserving reductions presented above, the inapproximability of $\DirEDPwC$ carries over to the respective $\VNEP$ variants. 

\begin{theorem}
\label{thm:beta-inapproximability}
Finding a $\beta$-approximate embedding for the $\VNEP$ variants $\var VE - $ and $\var E N $ is hard to approximate for $\beta \in \Theta(\log n / \log \log n)$, $n = |V_S|$, unless $\compNP \subseteq \mathcal{BP}\text{-}\textit{TIME}(\bigcup_{d \geq 1} n^{d \log \log n})$ holds.
\end{theorem}

\section{$\NPcompleteness$ under Graph Restrictions}

\label{sec:np-completeness-graph-restrictions}

All of our $\NPcompleteness$ results (except Theorem~\ref{thm:beta-inapproximability}) are based on a reduction from $\ThreeSAT$, yielding a specific directed-acylic substrate graph $\SG$ and a specific directed acyclic request graph $\VG$ and we note the following:
\begin{observation}
\label{obs:vnep-on-dags}
Theorems~\ref{thm:vnep-capacity-nodes-and-edges-is-np-complete} - \ref{thm:gamma-approximate-embeddings}  still hold when restricting the request and the substrate to acyclic graphs.
\end{observation}

Given the hardness of the $\VNEP$ and as for example Virtual Clusters (an undirected star network) can be optimally embedded in polynomial time~\cite{ccr15emb}, one might ask whether the hardness is preserved when restricting request graphs further.

In this section, we derive the result that the $\VNEP$ variants considered in this paper remain $\NPcomplete$ when request graphs are planar and degree-bounded. Our results are obtained by considering a \emph{planar} variant of $\ThreeSAT$. The planarity of a formula $\formula$ is defined according $\formula$'s graph interpretation:

\begin{definition}[Graph $\FG$ of formula $\formula$]
\label{def:graph-of-a-formula}
The graph $\FG=(\FV,\FE)$ of a SAT formula $\formula$ is defined as follows. $\FV$ contains a node $v_i$ for each clause $C_i \in \Cla$ and a node $u_k$ for each literal $x_k \in \Lit$. 
An undirected edge $\{v_i,u_k\}$ is contained in $\FE$, iff. the literal $x_k$ is contained in $C_i$ (either positive or negative).
Note that the graph $\FG$ is bipartite.
\end{definition}

An example for the interpretation $\FG$ is depicted in Figure~\ref{fig:3-sat-planar-illustration}.
Kratochv{\'\i}l~\cite{planar-3sat-special}  considered the following variant of $\SpecialThreeSAT$ and proved its $\NPcompleteness$.

\begin{definition/}[$\SpecialThreeSAT$]
\label{def:special-three-sat}
The 4-Bounded Planar 3-Connected 3-SAT ($\SpecialThreeSAT$) considers only $\ThreeSAT$ formulas $\formula$ for which the following holds:
\begin{description}
\item[\textnormal{(1)}] In each clause, exactly 3 distinct literals are used.
\item[\textnormal{(2)}] Each literal occurs in at most 4 clauses.
\item[\textnormal{(3)}] The graph $G_{\phi}$ is planar.
\item[\textnormal{(4)}] The graph $G_{\phi}$ is vertex 3-connected. \hspace*{\fill} $\square$
\end{description}
\end{definition/}

\begin{theorem}[\cite{planar-3sat-special}]
\label{thm:kratochvil-special-planar-3-sat}
$\SpecialThreeSAT$ is $\NPcomplete$. 
\end{theorem}

The following lemma connects $\SpecialThreeSAT$ formulas $\formula$ with the corresponding request graphs $\VG$.

\begin{figure}[tb!]
		\centering
\includegraphics[width=0.99\columnwidth]{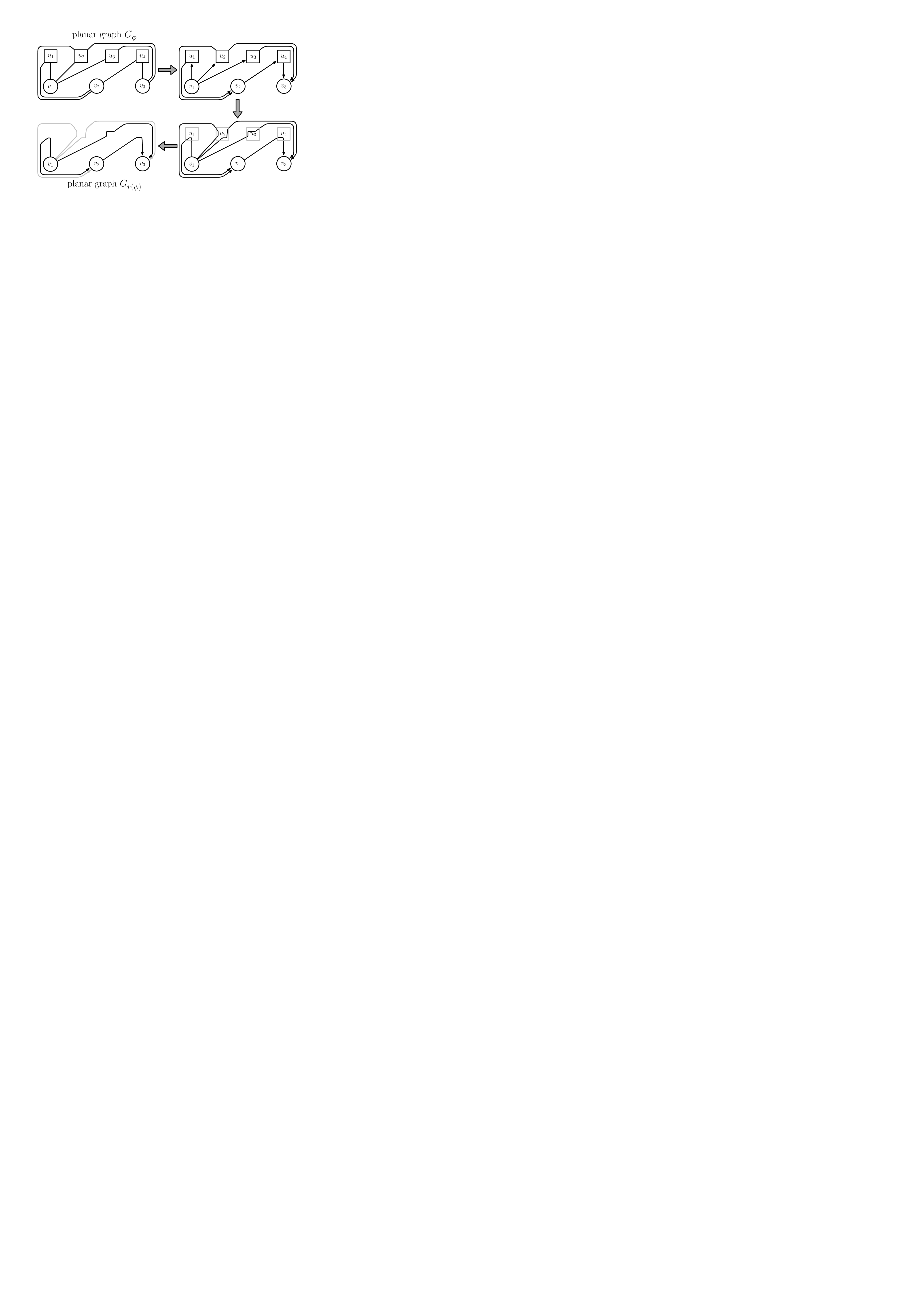}
\caption{Depicted is the transformation process of a planar graph $\FG$ \mbox{(cf. Definition~\ref{def:graph-of-a-formula})} to the planar graph $\VG$. Concretely, the example formula of Figure~\ref{fig:3-sat-illustration-new} is revisited, i.e. $\formula = C_1 \wedge C_2 \wedge C_3$, with $ C_1=x_1 \vee x_2 \vee x_3$, $C_2=\bar{x}_1 \vee x_2 \vee x_4$, and $C_3=x_2 \vee \bar{x}_3 \vee x_4$. 
In the first step all edges are directed, such that edges from clause nodes are oriented towards literal nodes iff. the literal occurs in the respective clause for the first time (according to the ordering of clause nodes). In the second step, each outgoing edge of a literal node is joined with the single incoming edge (duplicating it when necessary), hence allowing to remove the literal nodes. In the last step, duplicate edges are removed, yielding the request graph $\VG$. Each step of this transformation process safeguards the graph's planarity.}
\label{fig:3-sat-planar-illustration}
\end{figure}

\begin{lemma}
\label{lem:properties-of-request-graphs-special-three-sat}
Given a $\SpecialThreeSAT$ formula $\formula$, the following holds for the request graph $\VG$ (cf. Definition~\ref{def:request-construction}):
\begin{enumerate}
\item The request graph $\VG$ is planar.
\item The node-degree of $\VG$ is bounded by 12.
\end{enumerate}
\end{lemma}
\begin{proof}
We consider an arbitrary $\SpecialThreeSAT$ formula $\formula$ to which the conditions of Definition~\ref{def:special-three-sat} apply. We first show that the corresponding request graph $\VG$ is planar by detailing a transformation process leading from $\FG$ to $\VG$ while preserving planarity (see Figure~\ref{fig:3-sat-planar-illustration} for an illustration).

Starting with the undirected graph $\FG$, the edges are first oriented: an edge is oriented from a clause node to a literal node iff. the literal occurs in the respective clause for the first time according to the clauses' ordering. Note that while many reductions in Section~\ref{sec:np-completeness-of-the-vnep} required the reordering of clause nodes according to Lemma~\ref{lem:request-normalization}, this reordering preserves planarity as the structure of the graph $\FG$ does not change. 

Given this directed graph, the literal nodes are now removed by joining the single incoming edge of the literal nodes with \emph{each} outgoing edge of the corresponding literal node. In particular, consider the literal node $x_2$ of Figure~\ref{fig:3-sat-planar-illustration}: the single incoming edge $(C_1,x_2)$ is joined with the outgoing edges $(x_2,C_2)$ and $(x_2,C_3)$ to obtain the edges $(C_1,C_2)$ and $(C_1,C_3)$, respectively. As the duplication of the single incoming edge cannot refute planarity and all incoming and outgoing edges connect to the same node, the planarity of the graph is preserved in this step. Lastly, duplicate edges are removed to obtain the graph $\VG$, which is, in turn, planar.

It remains to show, that the request graph $\VG$ corresponding to $\formula$ exhibits a bounded node-degree of $12$ (in the undirected interpretation of the graph $\VG$). To see this, we note the following: based on the first two conditions of Definition~\ref{def:special-three-sat}, each clause node connects to exactly $3$ literal nodes and each literal node connects to at most $4$ clause nodes. Hence, when removing the literal nodes in the transformation process, the degree of each node may increase at most by a factor of $4$. As any clause node had $3$ neighboring literal nodes, this implies that the degree of any node is at most $12$ after the transformation process, completing the proof.
\end{proof}

Given the above, we easily derive the following theorem:

\begin{theorem}
\label{thm:np-completeness-on-restricted-request-graphs}
Theorems~\ref{thm:vnep-capacity-nodes-and-edges-is-np-complete} - \ref{thm:gamma-approximate-embeddings} hold when restricting the request graphs to be planar and / or degree 12-bounded. 
Theorem~\ref{thm:beta-inapproximability} holds for planar and degree 1-bounded graphs.
\end{theorem}
\begin{proof}
Our $\NPcompleteness$ proofs in Section~\ref{sec:np-completeness-of-the-vnep} and Section~\ref{sec:np-completeness-approximate-embeddings} (except for Theorem~\ref{thm:beta-inapproximability}) relied solely on the reduction from $\ThreeSAT$ to $\VNEP$ using the base Lemma~\ref{lem:base-lemma}. 
As formulas of $\SpecialThreeSAT$ are a strict subset of the $\ThreeSAT$ formulas, the base Lemma~\ref{lem:base-lemma} is still applicable for $\SpecialThreeSAT$ formulas. However, due to the structure of $\SpecialThreeSAT$ formulas, the corresponding requests in the reductions are planar and exhibit a node-degree bound of 12 by Lemma~\ref{lem:properties-of-request-graphs-special-three-sat}. Hence, solving the $\VNEP$ is $\NPcomplete$, even when restricting the requests to planar and / or degree-bounded ones.
Lastly, we note that Theorem~\ref{thm:beta-inapproximability} holds for planar and degree 1-bounded request graphs, as in the reduction only such requests were considered.
\end{proof}

\section{Conclusion}

We presented a comprehensive set of hardness results 
for the $\VNEP$ and its variants, which lie at the core of many resource allocation problems in networks.
Our results are negative in nature: we show that the
problem variants are $\NPcomplete$ and hence inapproximable (unless $\compPeqNP$) and that this holds true even for restricted classes of request graphs.

We believe that our results are of great importance for future work on several of the virtual network embedding problems. For example, our results on the variant enforcing node placement and latency restrictions are of specific interest for Service Function Chaining.
Surprisingly, the respective problem is hard even when not considering any capacity constraints. Furthermore, we have shown that it is hard to compute embeddings satisfying latency bounds within a factor of (less than) two times the original bounds.
In turn, whenever latency bounds must be obeyed strictly, one needs to rely on exact algorithmic techniques as e.g. Integer Programming.

\section*{Acknowledgements} This work was partially supported by Aalborg University's PreLytics project as well as by the German BMBF Software Campus grant 01IS1205.

}

\end{document}